\documentclass{article}
\usepackage{amsmath}
\usepackage{amssymb}
\usepackage[english]{babel}
\usepackage{amsthm}
\usepackage{tkz-graph}
\usepackage{bbm,mathtools}
\usepackage[authoryear,square]{natbib}
\usepackage{xcolor}
\usetikzlibrary{arrows, decorations.markings}

\tikzstyle{vecArrow} = [thick, decoration={markings,mark=at position
   1 with {\arrow[semithick]{open triangle 60}}},
   double distance=1.4pt, shorten >= 5.5pt,
   preaction = {decorate},
   postaction = {draw,line width=1.4pt, white,shorten >= 4.5pt}]
\tikzstyle{innerWhite} = [semithick, white,line width=1.4pt, shorten >= 4.5pt]
\newcommand{\op}[1]{\prescript{o}{}{#1}}

\newcommand{\otF}[1]{\prescript{\F}{}{#1}}
\newcommand{\oF}[1]{\prescript{\F}{}{#1}}
\newcommand{\oFX}[1]{\prescript{\F^X}{}{#1}}
\newcommand{\oG}[1]{\prescript{\G}{}{#1}}
\newcommand{\opQ}[1]{\prescript{Q,o}{}{#1}}
\newcommand{\optP}[1]{\prescript{\tilde P,o}{}{#1}}

\newcommand{\one}{\mathbbm 1}

\def\reals{\mathbb{R}}

\def\comp{\raise 1pt \hbox{$\scriptstyle\circ$}}

\def\upto{{\raise 1pt \hbox{$\scriptstyle \,\nearrow\,$}}}
\def\downto{{\raise 1pt \hbox{$\scriptstyle \,\searrow\,$}}}

\def\B{{\cal B}}
\def\C{{\cal C}}

\def\E{{\cal E}}
\def\bE{\mathbb E}
\def\F{{\cal F}}
\def\bF{{\mathbb F}}
\def\G{{\cal G}}
\def\bG{{\mathbb G}}

\def\bH{{\mathbb H}}

\def\M{{\cal M}}

\def\bN{\mathbb N}

\def\bP{{\mathbb P}}

\def\erf{\text{erf}}

\theoremstyle{plain}
\newtheorem{theorem}{Theorem}[section]
\newtheorem{lemma}[theorem]{Lemma}
\newtheorem{corollary}[theorem]{Corollary}
\newtheorem{proposition}[theorem]{Proposition}
\newtheorem{remark}[theorem]{Remark}
\theoremstyle{definition}

\newtheorem{example}[theorem]{Example}

\newtheorem{problem}[theorem]{Problem}
\newtheorem{notation}[theorem]{Notation}
\theoremstyle{empty}
\numberwithin{equation}{section}

\begin{document}
\title{Optional projection under equivalent local martingale measures}

\author{Francesca Biagini \and Andrea Mazzon \and Ari-Pekka Perkki\"o}

\maketitle

\begin{abstract}
Motivation for this paper is to understand the impact of information on asset price bubbles and perceived arbitrage opportunities. This boils down to study optional projections of $\bG$-adapted strict local martingales into a smaller filtration $\bF$ under equivalent martingale measures. We give some general results as well as analyze in details two specific examples given by the inverse three dimensional Bessel process and a class of stochastic volatility models.
\end{abstract}

\textbf{Keywords:}
 Local martingale, Optional projection, Local martingale measure, Filtration shrinkage, Bubbles

\textbf{AMS subject classications:}  	60G07, 60H30, 60G44

\section{Introduction}\label{sec:introduction}
The initial motivation for this paper is to study the impact of information on asset price bubbles and perceived arbitrage opportunities. In particular, we consider the case when a group of investors have access to restricted information. It is then an interesting question to ask whether these investors perceive the same bubbles seen by traders with full information, and whether they see \emph{illusionary} arbitrage opportunities,  as already investigated in \citep{jarrow2013positive} for a particular setting.

These questions boil down to study optional projections of $\bG$-adapted strict local martingales into a smaller filtration $\bF$ because, as we extensively explain in Section \ref{sec:setting}, bubbles are usually characterised as non-negative strict local martingales, and it is well known that the optional projection of a strict local martingale may fail to be a local martingale in the smaller filtration, see for example  \citep{follmer2011local} and \citep{larsson2014filtration}. Moreover, if the optional projection admits no equivalent local martingale measure, an illusionary arbitrage opportunity is perceived under the filtration $\bF$.

However, the impact of partial information cannot be disjoint to the choice of the underlying pricing measure, as the optional projection depends on the filtration and on the equivalent local martingale measure (ELMM in short in the sequel) with respect to which the conditional expectation is computed. For this reason, the study of the presence of bubbles as well as of perceived arbitrage opportunities under restricted information motivates an investigation of optional projections under changes of ELMMs.

In particular we study the relation among the set $\M_{loc}$ of ELMMs for a process $X$ and the set $\M_{loc}^o$ of measures $Q \sim P$ such that the optional projection under $Q$ is a $Q$-local martingale. We obtain full chacterization of the relations between $\M_{loc}$ and $\M_{loc}^o$: this allows to answer the original questions on asset price bubbles and arbitrages under partial information, as well as to provide new mathematical insights about optional projections. We focus on two main cases: the inverse three-dimensional Bessel process and an extension of the stochastic volatility model of \citep{sin1998complications}. We also consider the optional projection of a $\bG$-adapted process into the delayed filtration  $(\G_{t -\epsilon})_{t \ge 0}$, $\epsilon>0$: this is a case with interesting consequences for financial applications, as it represents the scenario where investors in the market have access to the information with a given positive time delay.

Moreover, we provide an invariance theorem about local martingales which are solution of a one-dimensional SDE in the natural filtration of an $n$-dimensional Brownian motion, see Theorem \ref{thm:invariance}. Specifically, we see that under mild conditions, such a local martingale $X$ has same law under $P$ as under every $Q \in \M_{loc}(X)$. This result is useful in our applications for two reasons: first, it implies that if $X$ is perceived as bubble under $P$, it is a bubble under any ELMM $Q$. Further, Theorem \ref{thm:solutionnolocmart}  gives a result about optional projections into a filtration $\bF \subseteq \bF^X$, with respect to an ELMM $Q$ such that $X$ has same law under $P$ as under $Q$: a class of local martingales $X$ such that all ELMM for $X$ have this characteristic is indeed provided by Theorem \ref{thm:invariance}. Important applications of the setting of Theorem \ref{thm:solutionnolocmart}, i.e., when $\bF \subseteq \bF^X$, are given by delayed information and by the model of \citep{cetin2004modeling}, where the market does not see the value of a firm but only knows when the firm has positive cash balances or when it has negative or zero cash balances.



The rest of the paper is organised as follows. In Section \ref{sec:setting} we motivate in details our analysis by presenting financial applications of optional projections with particular focus on some questions about financial asset bubbles and arbitrages under restricted information. We also describe our setting and summarize the aims of our study into five mathematical problems about optional projections of strict local martingales which we study in the sequel. Moreover, we here anticipate the main results of the paper.
In Section \ref{sec:genresults} we give some general results about optional projections of local martingales under equivalent local martingale measures, that will be used in Sections \ref{sec:bessel} and  \ref{sec:stocvol} in order to discuss the problems stated in Section \ref{sec:setting}.
More precisely, Section \ref{sec:bessel} is devoted to the inverse three-dimensional Bessel process, projected into different filtrations, whereas in Section  \ref{sec:stocvol} we focus on a class of two-dimensional stochastic volatility processes. We conclude the paper with Appendix \ref{app:opttransport}, where we characterize the martingale property of the  optional projection of a local martingale via optimal transport.  

\section{Motivation and setting}\label{sec:setting}

It is a well known fact that, when projecting a stochastic process into a filtration with respect to which it is not adapted, some basic properties can be lost. The most famous issue in this sense is that the optional projection of a local martingale may fail to be a local martingale, see for example \citep[Theorem 3.7]{follmer2011local} and \citep[Corollary 1]{larsson2014filtration}. 

When coming to financial applications, the possibly different nature of the optional projection implies that the market model with incomplete information (represented by the smaller filtration) may retain deeply different features with respect to the market model with full information.

As an example, an important financial implication of filtration shrinkage regarding credit risk is given by the models proposed in \citep{cetin2004modeling} and \citep{jarrow2007information}: taking inspiration from the work of \citep{jarrow2004structural}, the authors characterise reduced form models as optional projections of structural models into a smaller filtration. In particular, the cash balance of a firm, represented by a process $X=(X_t)_{t \ge 0}$, is adapted to the filtration $\bG$ of the firm's management, but not necessarily to the filtration $\bF$ representing the information available to the market. For this reason, the default time of the firm is a predictable stopping time for the firm's management looking at $\bG$, but it may be totally inaccessible for the market, based on $\bF$. Thus, the reduced information makes default a surprise to the market.

In \citep{jarrow2007information}, the authors introduce a zero-coupon bond issued by the firm under consideration, that promises to pay $1$ dollar at time $T$ and pays $\delta<1$ when defaulting at a time $\tau$ before $T$. The market price $v(t,T)$ of the bond at time $t<T $ depends on the expectation of $\one_{\{\tau <T\}}$ conditional to $\F_t$, under a martingale measure $P$ for $X$: this differs from the value of the bond $v^{mgmt}(t,T)$ computed by the firm's management, which depends on the expectation of $\one_{\{\tau <T\}}$ conditional to the larger information $\G_t$. In particular, if the spot rate is deterministic, we have $v(t,T)=\mathbb{E}^P[v^{mgmt}(t,T)|\F_t]$, that is, the market price of the bond is the optional projection of the value estimated by the management.

This feature of the model, where the optional projection of a process into the smaller filtration gives the traded market price of a financial product, may be extended besides credit risk: the management of a firm with a cash flow $X=(X_t)_{t \ge 0}$, for example, may be willing to hedge $X$ by selling a product in the market. However, it does not sell $X$ itself, but the optional projection $\op X$ of $X$ into the filtration $\bF$, which is typically smaller than the information $\bG$ containing also the value of the firm.

An interesting question to analyse is then if traders with partial information $\bF \subset \bG$ may perceive different characteristics of the market they observe with respect to the individuals with access to the complete information $\bG$. In particular, we focus on the perception of bubbles and possible arbitrages.

Bubbles are defined as the difference between the market wealth $W$ of an asset and its fundamental wealth $W^F$, usually seen as the expected sum of future discounted cash flows provided by the asset itself. In the context of the martingale theory of bubbles (see among the others \citep{CoxHobson}, \citep{LoewensteinWillard}, \citep{JarrowProtter2007},   \citep{JarrowProtter2010}, \citep{JarrowKchiaProtter}, \citep{Biagini}, \citep{Protter2013}), bubbles are non-negative, and in particular strictly positive if the price process of the asset is driven by a strict local martingale.

For this reason, an asset price process may be perceived as a bubble for a measure $Q$ under which $W$ is a strict local martingale, but not for a measure $R$ under which $W$ is a true martingale. To stress the fact that the perception of a bubble depends on the equivalent local martingale measure (ELMM) taken into consideration, which represents the pricing instrument chosen by the market, the concept of a $Q$-bubble is introduced, see \citep{Protter2013}. 

In the setting of a probability space $(\Omega,\F, P)$ endowed with a filtration $\bG=(\G_t)_{t \ge 0}$, the value of a $Q$-bubble at time $t$ is defined as 
$$
\beta^Q_t:=W_t-W^{F,Q}_t=W_t-\mathbb{E}^Q[W_T|\F_t], \quad t \ge 0,
$$
where $Q$ is an ELMM for $W$ and $T$ is the maturity or default time of the asset. 

The most interesting case is the one of incomplete markets (see \citep{JarrowProtter2010}, \citep{Protter2013}, \citep{Biagini}) where infinitely many ELMMs exist: here, the birth and the evolution of a bubble are then determined by a flow of different ELMMs that gives rise to a corresponding shifting perception of the fundamental value of the asset.

However, when coming to a setting where a dichotomy is present between a process $X$ adapted to a filtration $\bG$ and the optional projection $\op X$ into a filtration $\bF \subset \bG$, the bubble is not only characterised by the view of the market, identified by the ELMM, but also by the level of information, identified by the filtration. 

Let us consider our example above, and suppose that $X$ is a strict local martingale under the given measure $P$, so that the management of the firm with cash flow given by $X$ perceives a strictly $(P,\bG)$-bubble defined by
$$
\beta^{(P,\bG)}_t:=X_t-\mathbb{E}^P[X_T|\G_t], \quad t \ge 0.
$$
Consider then investors in the market with access to the restricted information $\bF$. If the optional projection $\op X$ under $P$ is a $(P,\bF)$-local martingale, then traders perceive a $(P,\bF)$-bubble
\begin{align}
\beta^{(P,\bF)}_t&:=\op X_t-\mathbb{E}^P[\op X_T|\F_t] = \op X_t-\mathbb{E}^P\left[\mathbb{E}^P[X_T|\F_T]|\F_t\right]  \notag \\ & = \mathbb{E}^P[X_t - X_T|\F_t], \quad t \ge 0.\label{eq:PFbubble}
\end{align}
However, it is well known that it is not always the case that the optional projection of a strict local martingale is a local martingale, see \citep[Theorem 3.7]{follmer2011local} and \citep[Corollary 1]{larsson2014filtration}. When the local martingale property is not preserved with respect to the optional projection under $P$, then $\beta^{(P, \bF)}$ in \eqref{eq:PFbubble} is not a $P$-bubble process. 

It is then natural to ask whether there exists an ELMM $Q$ for $X$ such that the $Q$-optional projection defined by $\opQ X_t = \mathbb{E}^Q[X_t|\F_t]$ is a $(Q,\bF)$-local martingale. In this case, indeed, it would be possible to define the $(Q,\bF)$-bubble 
\begin{align}\label{eq:QFbubble}
\beta^{(Q,\bF)}_t&:=\opQ X_t-\mathbb{E}^Q[\opQ X_T|\F_t] = \mathbb{E}^Q[X_t - X_T|\F_t], \quad t \ge 0,
\end{align}
and compare it with the $(Q,\bG)$-bubble
$$
\beta^{(Q,\bG)}_t:=X_t-\mathbb{E}^Q[X_T|\G_t], \quad t \ge 0.
$$
In particular, an interesting question is if both the management of the firm and the traders in the market perceive a strictly positive bubble under $Q$. This is the case if $X$ is a $(Q,\bG)$-strict local martingale and $\opQ X$ is a $(Q,\bF)$-strict local martingale. 

Other issues arise when considering the possible existence of perceived arbitrages. As already noted in \citep{jarrow2013positive}, traders with limited information may interpret the bubble's impact on the price process as an arbitrage opportunity. This happens if $X$ is a $(P,\bG)$-strict local martingale, $\op X$ fails to be a $(P,\bF)$-strict local martingale and, in addition, there exists no measure $Q \sim P$ under which $\op X$ is a local martingale. 

To summarise, four main questions may be asked about the possibly different perception of bubbles and arbitrage opportunities under limited information:
\begin{enumerate}
\item Does there exist an ELMM measure $Q$ for the original process $X$ such that $\beta^{(Q,\bF)}$ defined in \eqref{eq:QFbubble} is a bubble process with respect to $(Q,\bF)$?
\item Does there exist an ELMM measure $Q$ such that both the management of the firm with access to $\bG$ and traders with access to $\bF$ perceive a strictly positive bubble under $Q$?
\item Can a bubble under the probability $P$ and  filtration $\bG$ be perceived as an arbitrage opportunity under $P$ and $\bF$?
\item Is it  the case for every ELMM $Q$, that what is perceived as a bubble under the probability $Q$ and  filtration $\bG$, is perceived as an arbitrage opportunity under $Q$ and $\bF$?
\end{enumerate}

Goal of the present paper is to give some insights on the questions presented above, by providing an answer to five related mathematical problems about optional projections and ELMM we introduce in Section \ref{subsec:setting}.

\subsection{Mathematical setting}\label{subsec:setting}

Consider a probability space $(\Omega, \F, P)$ equipped with two filtrations $\bF=(\F_t)_{t \ge 0}$, $\bG=(\G_t)_{t \ge 0}$, satisfying the usual hypothesis of right-continuity and completeness, with $\bF \subset \bG$. Moreover, let $X$ be a positive c\`adl\`ag $(P,\bG)$-local martingale. Unless differently specified, we suppose in particular $X$ to be a strict $(P,\bG)$-local martingale.

For the rest of the paper, we adopt the following notation.

\begin{notation}
We denote by $\op X$ the optional projection of $X$ into $\bF$, i.e., the unique c\`adl\`ag process satisfying 
$$
\one_{\{\tau<\infty\}}\op X_{\tau} = \bE[\one_{\{\tau<\infty\}}X_{\tau}|\F_{\tau}] \quad a.s.,
$$
for every $\bF$-stopping time $\tau$. We also define $\opQ X$ to be the optional projection of $X$ under $Q \sim P$ into $\bF$, i.e.,  $(\opQ X)_t = \opQ X_t := \bE^Q[X_t|\F_t]$. We call $\opQ X$ the $Q$-optional  projection of $X$, whereas if we don't specify the measure, the optional projection is always meant to be with respect to $P$.

We call $\bF^X$ the natural filtration of $X$. Moreover, if $Q$ is a probability measure equivalent to $P$, we define $Z_{\infty}:=\frac{dQ}{dP}$ and denote by $\oF Z$, $\oFX Z$, $\oG Z$, the c\`adl\`ag processes characterised by 
\begin{equation}\label{eq:densities}
\oF Z_t = \bE[Z_{\infty}|\F_t], \quad \oFX Z_t = \bE[Z_{\infty}|\F^X_t],  \quad \oG Z_t = \bE[Z_{\infty}|\G_t], \quad t \ge 0,
\end{equation}
respectively. Moreover, for $\bH = \bF, \bF^X, \bG$, we denote
\begin{align}
\M_{loc}(X,\bH)&=\{Q \sim P, \quad \text{X is a $(Q,\bH)$-local martingale}\}, \notag \\
\M_{M}(X,\bH)&=\{Q \sim P, \quad \text{X is a $(Q,\bH)$-true martingale}\}, \notag \\
 \M_{L}(X,\bH)& =\{Q \sim P, \quad \text{X is a $(Q,\bH)$-strict local martingale}\}.\notag
\end{align}
We also set
{\small{
\begin{align}\notag
\M^o_{loc}(X, \bF):&=  \left\{Q \sim P, \quad \text{$\opQ X$ is a $(Q, \bF)$-local martingale} \right\}, \notag\\
\M_{loc}(X, \bG, \bF):&= \left\{Q \sim P, \quad \text{$X$ is a $(Q,\bG)$-local martingale}, \quad \left(\frac{dQ}{dP}|_{\G_t}\right)_{t \ge 0}  \text{ is $\bF$-adapted}\right\},\notag \\
\M^o_{loc}(X, \bG, \bF):&=  \left\{Q \sim P, \quad \text{$\opQ X$ is a $(Q, \bF)$-local martingale}, \quad \left(\frac{dQ}{dP}|_{\G_t}\right)_{t \ge 0}  \text{ is $\bF$-adapted}\right\}. \notag
\end{align}
}}
\end{notation}



We investigate the following properties of the model, related to the four questions formulated in Section \ref{sec:setting}. More precisely, we study when:
\begin{align}
& \M_{loc}(X,\bG) \cap \M^o_{loc}(X, \bF) \ne \emptyset;  \tag{P1} \label{prob:locmart} \\
& \M_{loc}(\op X,\bF) \ne \emptyset;\tag{P2} \label{prob:locmartgen}\\
& \M_{L}(X,\bG) \cap \M^o_{loc}(X, \bF) \ne \emptyset; \tag{P3}\label{prob:intersect}\\
& \M_{loc}(X, \bG, \bF)=\M^o_{loc}(X, \bG, \bF); \tag{P4} \label{prob:equal}\\
& \bigcup\limits_{Q \in \M_{loc}(X,\bG)} \M_{loc}(\opQ X, \bF) \ne \emptyset. \tag{P5}\label{prob:union}
\end{align}

%
%
%
%
%


Note that $\M_{L}(X,\bG), \M_{loc}(X,\bG) \ne \emptyset$, as $P \in \M_{L}(X,\bG)$ by hypothesis. Properties \eqref{prob:locmart}, \eqref{prob:locmartgen}, \eqref{prob:intersect} and \eqref{prob:union} trivially hold if $\op X$ is an $\bF$-local martingale, so the more interesting case is when the $P$-optional projection is not a local martingale. Under this hypothesis, properties \eqref{prob:locmart}-\eqref{prob:union} can hold or not depending on both the process $X$ and the filtration $\bF$, as illustrated in Sections \ref{sec:bessel} and \ref{sec:stocvol}. In particular, they are all related as we also discuss in the sequel. 

Note that if one of \eqref{prob:intersect} or \eqref{prob:equal} holds, \eqref{prob:locmart} is also true, and that \eqref{prob:locmart} trivially implies  \eqref{prob:locmartgen} . Moreover, the property \eqref{prob:union} is the weakest one: if any of \eqref{prob:locmart}, \eqref{prob:locmartgen}, \eqref{prob:intersect} or \eqref{prob:equal} holds, this implies that \eqref{prob:union} is true. 

This can be summarized in the following scheme:
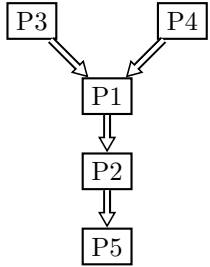
\begin{center}
\begin{tikzpicture}[thick]
  \node[draw,rectangle] (P3) {P3};
  \node[inner sep=0,minimum size=0,right of=P3] (k) {}; 
  \node[draw,rectangle,right of=k] (P4) {P4};
  \node[draw,rectangle,below of=k] (P1) {P1};
  \node[draw,rectangle,below of=P1] (P2) {P2};
  \node[draw,rectangle,below of=P2] (P5) {P5};

  \draw[vecArrow] (P3) to (P1);
  \draw[vecArrow] (P4) to (P1);
  \draw[vecArrow] (P1) to (P2);
  \draw[vecArrow] (P2) to (P5);

\end{tikzpicture}
\end{center}

Properties \eqref{prob:locmart}-\eqref{prob:union} are strictly related to the questions introduced in Section \ref{sec:setting}. In particular, \eqref{prob:locmart}  implies the existence of a measure $Q \in \M_{loc}(X,\bG) \cap \M^o_{loc}(X, \bF)$, so that it is possible to define a $(Q,\bF)$-bubble as in \eqref{eq:QFbubble}. Moreover, if \eqref{prob:equal} holds, this implies that the measures $Q \sim P$ for which the $(Q,\bF)$-bubble in \eqref{prob:equal} is well defined and such that $(dQ/dP|_{\G_t})_{t \ge 0}  \text{ is $\bF$-adapted}$, are  the ELMMs for $X$ such that $(dQ/dP|_{\G_t})_{t \ge 0}  \text{ is $\bF$-adapted}$.

However, if we take a general measure  $Q \in \M_{loc}(X,\bG) \cap \M^o_{loc}(X, \bF)$, the $(Q,\bG)$-bubble as well as the $(Q,\bF)$-bubble can be zero. But if  \eqref{prob:intersect} holds, this means that there exists an ELMM $Q$ for $X$ such that a strictly positive $Q$-bubble is perceived both under $\bG$ and $\bF$, providing a positive answer to the second question of Section \ref{sec:setting}. 

Moreover, the third and the fourth question of Section \ref{sec:setting}, related to the perception of arbitrages under the smaller filtration $\bF$, are investigated in \eqref{prob:locmartgen} and \eqref{prob:union}, respectively.


In the rest of this section we anticipate the results we obtain. Note that properties \eqref{prob:locmart}, \eqref{prob:locmartgen}, \eqref{prob:intersect} and  \eqref{prob:union} trivially hold for the three-dimensional Bessel process projected into the filtration generated by $B^1$, $B^2$, as the optional projection is again a local martingale, see Section \ref{subsec:B1B2}. 

\begin{itemize}
\item \textbf{Property \eqref{prob:locmart}}: in Section \ref{sec:stocvol} we introduce a stochastic volatility process $S$, which is a strict local martingale under suitable conditions on the coefficients of its SDE, but whose optional projection into a specific sub-filtration is not a local martingale, see Theorem \ref{thm:nolocalized}. Property \eqref{prob:locmart} holds because $S$ admits a true martingale measure, see Proposition \ref{prop:sin}. 
On the contrary, \eqref{prob:locmart} is not true for the inverse three-dimensional Bessel process projected into a delayed filtration, i.e., $\bF = (\F_t)_{t \ge 0}$, $\F_t = \G_{t - \epsilon}$, $\epsilon>0$, see Remark \ref{rem:delayedextended}.

\item \textbf{Property \eqref{prob:locmartgen}}: a particular case of the stochastic volatility model introduced in Section \ref{sec:stocvol} permits to construct a strict local martingale  $X$ such that $\op X$ is not a local martingale but $\M_{loc}(\op X,\bF) \ne \emptyset$, see Example \ref{ex:problem2}, so that property \eqref{prob:locmartgen} holds. This is also the case for the optional projection into the delayed filtration of the process introduced in Example \ref{ex:invbesslminusbm}. 
For the inverse three-dimensional Bessel process projected into the delayed filtration, on the other hand, \eqref{prob:locmartgen} is not true, see Theorem \ref{thm:Mlocempty}.

\item \textbf{Property \eqref{prob:intersect}}: taking the stochastic volatility process of Section \ref{sec:stocvol} and adding it to a suitable strict local martingale, we get a strict local martingale $X$ and a filtration $\bF$ such that the optional projection of $X$ into $\bF$ is not a local martingale and \eqref{prob:intersect} holds, see Example \ref{ex:problem3}. On the contrary, this property is not true for the inverse three-dimensional Bessel process projected into the delayed filtration, see Section \ref{subsec:delayed}.

\item \textbf{Property \eqref{prob:equal}}: it holds for the  inverse three-dimensional Bessel process projected into the filtration generated by $B^1$ and $B^2$, see Theorem \ref{thm:invbessel12}. It is not true for all the examples when property \eqref{prob:locmart} does not hold, e.g., in the case of the inverse three-dimensional Bessel process projected into a delayed filtration.

\item \textbf{Property \eqref{prob:union}}: this is true for all the examples considered except for the inverse three-dimensional Bessel process projected into the delayed filtration, see Theorem \ref{thm:Mlocempty}.

\end{itemize}

 \section{General results}\label{sec:genresults}
 

We now give some general results about optional projections under changes of equivalent measures, in the setting introduced in Section \ref{subsec:setting}. Some of the examples that  we provide in Sections \ref{sec:bessel} and \ref{sec:stocvol} are based on these findings.

The following theorem provides a condition under which the $Q$-optional projection of $X$ is an $\bF$-local martingale under any ELMM $Q$.

\begin{theorem}\label{thm:samesequence}
 Assume that $X$ admits an $\bF$-localizing sequence which makes it a bounded $(P,\bG)$-martingale. Then $\opQ X$ is a $(Q, \bF)$-local martingale for every $Q\in \M_{loc}(X,\bG)$.
\end{theorem}
\begin{proof}
Let $Q\in \M_{loc}(X,\bG)$, and $(\tau_n)_{n \in \bN}$ be the assumed localizing sequence. Since $X^{\tau_n}$ is bounded for every $n \in \bN$, $(\tau_n)_{n \in \bN}$ localizes $X$ under $Q$ as well, and the result follow  by Theorem 3.7 in \citep{follmer2011local}.
\end{proof}

We now give a theorem which provides a class of local martingales whose law under $P$ is invariant under change of any equivalent local martingale measure. This result is of independent interest and also useful in our context, see 
Theorem \ref{thm:solutionnolocmart}.

\begin{theorem}\label{thm:invariance}
Let $\bG$ be the natural filtration of an $n$-dimensional Brownian motion $B=(B_t)_{t \ge 0}$, $n \in \mathbb{N}$. Moreover, let $X=(X_t)_{t \ge 0}$ be a $(P,\bG)$-local martingale, unique strong solution of the SDE 
\begin{align}
dX_t & = \sigma(t,X_t)dW_t,\quad  t \ge 0, \label{eq:Xforinv} 
\end{align} 
where $W$ is a one-dimensional $(P,\bG)$-Brownian motion, and the function $\sigma(\cdot, \cdot)$ is such that there exists a unique strong solution to \eqref{eq:Xforinv}. Suppose also that $\sigma(t,X_t) \ne 0$ a.s. for almost every $t \ge 0$.

Thus $X$ has the same law under $P$ as under any $Q \in \M_{loc}(X,\bG)$. In particular,  if $X$ is a $(P,\bG)$-strict local martingale, it is a $(Q,\bG)$-strict local martingale under any $Q \in \M_{loc}(X,\bG)$, and if it is a $(P,\bG)$-true martingale, it is a $(Q,\bG)$-true martingale under any $Q \in \M_{loc}(X,\bG)$.
\end{theorem}
\begin{proof}
The Martingale representation Theorem applied to the filtration $\bG$ implies that there exists a unique $\reals^n$-valued process  $\sigma^W= (\sigma^W_t)_{t \ge 0}$, progressive and such that $\int_0^t (\sigma^W_s)^2ds <\infty$ a.s. for all $t \ge 0$, such that
\begin{equation}\label{eq:reprW}
W_t = \int_0^t \sigma^W_s \cdot dB_s, \quad a.s., \quad t \ge 0.
\end{equation}
Consider a probability measure $Q \in \M_{loc}(X,\bG)$, defined by a density
\begin{equation}\label{eq:densityQ2}
\frac{dQ}{dP}|_{\G_t}=\mathcal{E}\left(\int_0^{t} \alpha_s \cdot dB_s\right), \quad t \ge 0,
\end{equation}
where $\alpha=(\alpha_t)_{t \ge 0}$ is a suitably integrable, $\bG$-adapted processes. Girsanov's Theorem implies that the dynamics of $W$ under $Q$ are given by
\begin{align}
W_t & = \int_0^t \sigma^W_s \cdot d \tilde B_s + \int_0^t (\sigma^W_s \cdot \alpha_s) ds, \quad  t \ge 0, 
\end{align}
where the process $\tilde B = (\tilde B_t)_{t \ge 0}$ defined by
 \begin{align}
\tilde B_t  = B_t - \int_0^t \alpha_s ds, \quad t \ge 0,
\end{align}
is a $(Q,\bG)$-Brownian motion. But since $Q \in \M_{loc}(X,\bG)$ and $\sigma(t,X_t) \ne 0$ a.s. for almost every $t \ge 0$, by equation \eqref{eq:Xforinv} we obtain that
$$
\sigma^W_t \cdot \alpha_t = 0, \quad t \ge 0,
$$
and $W$ is a $(Q,\bG)$-local martingale. 

Since $W$ is a $(P,\bG)$-Brownian motion, L\'evy's Characterization Lemma of the one-dimensional Brownian motion implies that $W$ is also a $(Q,\bG)$-Brownian motion, and the result follows.
\end{proof}


The next results regard the case when $P$ is the unique equivalent local martingale measure for $X$ with respect to $\bF^X$, i.e., when there are no ELMMs defined by a non trivial density adapted to $\bF^X$. In particular, Corollary \ref{cor:sameprojection} provides a proof to the fact that property \eqref{prob:locmart} does not hold in the setting of Remark \ref{rem:delayedextended}.

\begin{proposition}\label{prop:projectionGZ1}
Let $X$ be a $(P,\bG)$-local martingale, and suppose that  
$$
\M_{loc}(X, \bF^X) = \{P\}.
$$
Let $Q$ be a probability measure with $Q \in \M_{loc}(X,\bG)$, and $\oG Z$ be the density process defined in \eqref{eq:densities}.
Then it holds
\begin{equation}\notag
\bE[\oG Z_t|\F^X_t] = 1, \quad a.s., \quad t \ge 0.
\end{equation}
\end{proposition}
\begin{proof}
If $Q \in \M_{loc}(X,\bG)$, it follows that $X$ is  also a $(Q,\bF^X)$-local martingale since it is obviously adapted to $\bF^X$.
This implies that $X \cdot \oFX Z$ is a $(P,\bF^X)$-local martingale, where $\oFX Z$ is defined in \eqref{eq:densities}. By the hypothesis $\M_{loc}(X, \bF^X) = \{P\}$, for every $t \ge 0$ we have
\begin{equation}\notag
1 = \bE[Z_{\infty}|\F^X_t] = \bE\left[\bE[Z_{\infty}|\G_t] |\F^X_t\right] = \bE[\oG Z_t|\F^X_t], \quad a.s., \quad t \ge 0.
\end{equation}
\end{proof}
\begin{corollary}
Let $X$ be a $(P,\bG)$-local martingale, and suppose that  
$$
\M_{loc}(X, \bF^X) = \{P\}.
$$
Thus if $X$ is a $(P,\bG)$-strict local martingale, it is a $(Q,\bG)$-strict local martingale under any $Q \in \M_{loc}(X,\bG)$, and if it is a $(P,\bG)$-true martingale, it is a $(Q,\bG)$-true martingale under any $Q \in \M_{loc}(X,\bG)$.
\end{corollary}
\begin{proof}
Suppose that $X$ is a $(P,\bG)$-strict local martingale. Thus there exists $t \ge 0$ such that $\mathbb{E}^P[X_t]<X_0$, i.e., $X$ loses mass at some point. For the same $t$, we have 
\begin{align}
\bE^Q[X_t]&=\bE^P[Z_{\infty}X_t] = \bE^P\left[\bE^P[Z_{\infty}X_t|\G_t]\right] = \bE^P[\oG Z_tX_t] =  \bE^P\left[\bE^P[\oG Z_tX_t|\F^X_t]\right] \notag \\
&=\bE^P\left[X_t\bE^P[\oG Z_t|\F^X_t]\right]=\bE^P\left[X_t\right]<X_0,\quad a.s., \notag 
\end{align}
where the last equality follows from Proposition \ref{prop:projectionGZ1}. Thus $X$ loses mass at some point under $Q$ as well and it is therefore a $(Q,\bG)$-strict local martingale. Analogously, it can be seen that if $X$ is a $(P,\bG)$-true martingale, it is a $(Q,\bG)$-true martingale.
\end{proof}

\begin{corollary}\label{cor:sameprojection}
Let $X$ be a $(P,\bG)$-local martingale, and suppose that  
$$
\M_{loc}(X, \bF^X) = \{P\}.
$$
Thus for every probability measure $Q \in \M_{loc}(X,\bG)$, and every sub-filtration $\bF \subseteq \bF^X$, it holds $\op X = \opQ X$ a.s., i.e.,
$$
\bE^P[X_t|\F_t]=\bE^Q[X_t|\F_t], \quad a.s., \quad t \ge 0.
$$
\end{corollary}
\begin{proof}
Let $Q \in \M_{loc}(X,\bG)$. Thus we have
\begin{align}
\bE^Q[X_t|\F_t] &= \left(\bE^P[Z_{\infty}|\F_t]\right)^{-1} \bE^P[\oG Z_tX_t |\F_t] \notag\\
&=  \left(\bE^P\left[\bE^P[Z_{\infty}|\F^X_t]|\F_t\right]\right)^{-1} \bE^P\left[\bE^P[\oG Z_tX_t |\F^X_t]|\F_t\right]\notag\\
&=  \bE^P\left[X_t\bE^P[\oG Z_t |\F^X_t]|\F_t\right]\notag\\
&=  \bE^P\left[X_t|\F_t\right]\notag, \quad a.s., \quad t \ge 0,
\end{align}
where the second equality follows from the assumption that $\bF \subseteq \bF^X$ and the third and last equalities follow from Proposition \ref{prop:projectionGZ1}.
\end{proof}

 \subsection{Application to the equivalent measure extension problem}\label{subsec:larsson}
 
 As part of our results is based on the equivalent measure extension problem of \citep{larsson2014filtration}, we briefly recall it in the following, together with the most important results relating this problem to the optional projection of strict local martingales.
 
In the setting introduced in Section \ref{subsec:setting}, define first the stopping times  
$$
\tau_n := n \wedge \inf \{t \ge 0: X_t \ge n \}, \qquad \tau := \lim_{n \to \infty} \tau_n,
$$ 
and note that $\G_{\tau-}=\cup_{n\ge 1}\G_{\tau_n}$. 
\\
The F\"ollmer measure $Q_0$ is defined on $\G_{\tau-}$ as the probability measure  that coincides with $Q_n$ on $\G_{\tau_n}$ for each $n \ge 1$, where $Q_n \sim P$ is defined on $\G_{\tau_n}$ by $d Q_n = X_{\tau_n} dP$. For more details see \citep[Section 2]{larsson2014filtration}. \\
The new measure $Q_0$ is then only defined on $\G_{\tau-}$. It is then a natural question whether $Q_0$ can be extended to $\G_{\infty}$, i.e., of it is possible to find a measure $\tilde Q$ on $(\Omega, \G_{\infty})$ such that $\tilde Q=Q_0$ on $\G_{\tau-}$. There are several ways in which $Q_0$ can be extended to a measure $\tilde Q$ on $\G_{\infty}$, see \citep{larsson2014filtration}.
 A further problem is whether $Q_0$ admits an extension to $\G_{\infty}$ as specified below.
 \begin{problem}[Equivalent measure extension problem, Problem 1 of \citep{larsson2014filtration}]
Given the probability measure $Q_0$ introduced above, and two filtrations $\bF \subset \bG$, find a probability measure $Q$ on $(\Omega, \G_{\infty})$ such that:
\begin{enumerate}
\item $Q=Q_0$ on $\G_{\tau-}$;
\item The restrictions of $P$ and $Q$ to $\F_t$ are equivalent for each $t \ge 0$.
\end{enumerate}
\end{problem}
The existence of a solution to the equivalent measure extension problem is connected with the behaviour of the optional projection of X into $\bF$ by the following theorem. 
\begin{theorem}[Corollary 1 of \citep{larsson2014filtration}]\label{thm:larsson}
If $\op X$ is an $\bF$-local martingale, then the equivalent measure extension problem has no solution.
\end{theorem}

We now provide a result about optional projections under equivalent local martingale measures into a filtration $\bF  \subseteq \bF^X$. 
We start with a lemma.

\begin{lemma}\label{lem:solGFX}
Suppose that  the equivalent measure extension problem admits a solution for $P$ and the two filtrations $\bF \subseteq \bG$, with $\bF \subseteq \bF^X$. Then it also admits a solution for $P$, $\bF$ and $\bF^X$.
\end{lemma}
\begin{proof}
Call $Q$ a solution of the equivalent measure extension problem for $P$, $\bF$ and $\bG$, and let $Q_0$ and $Q_0^X$ be the F\"ollmer measures on $\G_{\tau-}$ and $\F^X_{\tau-}$, respectively. By construction, we have that $Q_0^X$ coincides with $Q_0$ on $\F^X_{\tau-}$.  This implies that $Q$ is also an extension of $Q_0^X$, equivalent to $P$ on $\F_t$ for every $t \ge 0$. Then $Q$ gives a solution for the equivalent measure extension problem for  $P$, $\bF$ and $\bF^X$.
\end{proof}

\begin{theorem}\label{thm:solutionnolocmart}
Consider a probability measure $\tilde P \in \M_{loc}(X,\bG)$ and suppose that $X$ has same law under $P$ as under $\tilde P$. Also assume that the equivalent measure extension problem admits a solution for $P$, and that $\bF  \subseteq \bF^X$. Thus the $\tilde P$-optional projection $\optP X$ of $X$ into $\bF$ is not a $(\tilde P,\bF)$-local martingale. 
\end{theorem} 
\begin{proof}
By Lemma \ref{lem:solGFX}, we have that the equivalent measure extension problem admits a solution for $P$, $\bF$ and $\bF^X$. Consider now the construction of the F\"ollmer measure illustrated above. The stopping times
$$
\tau_n := n \wedge \inf \{t \ge 0: X_t \ge n \}, \qquad \tau := \lim_{n \to \infty} \tau_n,
$$ 
have same law under $P$ as under $\tilde P$. Moreover, since they are defined by $d Q_n = X_{\tau_n} dP$ and $d \tilde Q_n = X_{\tau_n} d\tilde P$, the measures $Q_n$ and $\tilde Q_n$ coincide on $\bF_{\tau_n}^X$, so the the equivalent measure extension problem also admits a solution for $\tilde P$, $\bF$ and $\bF^X$. By Theorem \ref{thm:larsson}, it follows that the $\tilde P$-optional projection of $X$  into $\bF$  is not a $(\tilde{P},\bF)$-local martingale. 
\end{proof}
Note that Theorem \ref{thm:invariance} implies that Theorem \ref{thm:solutionnolocmart} can be applied to all processes with dynamics given by \eqref{eq:Xforinv}. An important application when $\bF \subseteq \bF^X$ is the study of delayed information.

We now discuss properties \eqref{prob:locmart}-\eqref{prob:union} in two specific cases in the following Sections \ref{sec:bessel} and \ref{sec:stocvol}.
\section{The inverse three-dimensional Bessel process}\label{sec:bessel}
Let $B^1=(B^1_t)_{t \ge 0}, B^2= (B^2_t)_{t \ge 0}, B^3=(B^3_t)_{t \ge 0}$ be standard, independent Brownian motions, starting at $(B^1_0,B_0^2,B_0^3)=(1,0,0)$, on $(\Omega, \F, P)$. We specify the filtration later. The  inverse three-dimensional Bessel process $M=(M_t)_{t \ge 0}$ is defined by
\begin{equation}\label{eq:bessel}
M_t:=\left((B_t^1)^2+(B_t^2)^2+(B_t^3)^2\right)^{-1/2}, \quad t \ge 0.
\end{equation}
It\^o's formula implies that under the original probability measure $P$, $M$ has dynamics
\begin{equation}\label{Mdynamics}
dM_t = - M_t^3 \left(B_t^1dB^1_t + B^2_t d B_t^2 + B_t^3dB^3_t\right), \quad t \ge 0,
\end{equation}
with $(B^1_0,B_0^2,B_0^3)=(1,0,0)$. It can be noted that $M$ also solves the SDE
\begin{equation}\label{eq:Mnewdynamics}
dM_t = - M_t^2 dW_t, \quad t \ge 0,
\end{equation}
where the process $W$ with
\begin{equation}\label{eq:W}
W_t=\int_0^t M_s \left(B_s^1dB^1_s + B^2_s d B_s^2 + B_s^3dB^3_s\right),\quad t \ge 0,
\end{equation}
 is a one-dimensional Brownian motion as it is a continuous local martingale with $[W,W]_t=t$. 
 
In this section we consider two different choices for the filtration $\bG$: in Section
 \ref{subsec:B1B2} we let $\bG$ be the filtration generated by $B^1$, $B^2$ and $B^3$, whereas in Section \ref{subsec:delayed}, $\bG$ is generated by the Brownian motion $W$ in  \eqref{eq:W}. In both cases, $M$ is a strict $\bG$-local martingale, and it is therefore interesting to investigate properties \eqref{prob:locmart}-\eqref{prob:union} when $M$ is projected into a smaller filtration $\bF$. In particular, in Section
  \ref{subsec:B1B2} we consider the case when $\bF$ is generated by $B^1$ and $B^2$. On the other hand, in Section \ref{subsec:delayed} we study an example of delayed information, which describes in fact a situation which often happens in practice: here $\bF=(\F_t)$, with $\F_t=\G_{t-\epsilon}$, $\epsilon>0$, meaning that investors have access to the information of the process with a strictly positive time delay $\epsilon$.

\begin{remark}\label{remark:completeincomplete}
In order to study property \eqref{prob:locmart}, it is of course important to have some knowledge about the set $\M_{loc}(M,\bG)$. In particular, one can ask if the market is complete, i.e., $\M_{loc}(M,\bG) = \{P\}$, or if it is incomplete, that means that there exists infinitely many measures $Q \in \M_{loc}(M,\bG)$. 

In the case of the inverse three-dimensional Bessel process, this depends on the choice of $\bG$: if $\bG$ is generated by one Brownian motion, as it happens in Section \ref{subsec:delayed}, it is well known that the market is complete, so that the probability $P$ is the only measure under which $M$ is a local martingale, see also \citep{delbaen1994arbitrage}.

On the other hand, let now $\bG$ be the natural filtration of $B^1$, $B^2$ and $B^3$, as it is the case in Section
\ref{subsec:B1B2}. In this case, $\M_{loc}(M,\bG) \ne \emptyset$ but $\M_{loc}(M,\bG) \ne \{P\}$. Namely, consider for example the $\bG$-adapted processes 
\begin{equation}\label{eq:lineardependence}
\alpha_t^1 = -\frac{B_t^2}{(B_t^1)^2+(B_t^2)^2+1}, \qquad \alpha_t^2 = \frac{B_t^1}{{(B_t^1)^2+(B_t^2)^2+1}}, \qquad \alpha^3_t = 0, \quad t \ge 0,
\end{equation}
and $Z=(Z_t)_{t \ge 0}$ defined by
\begin{equation}\label{eq:densityQ}
Z_t=\mathcal{E}\left(\int_0^{t} \alpha_s^1dB_s^1+\int_0^{t} \alpha_s^2dB_s^2+\int_0^{t} \alpha_s^3dB^3_s\right), \quad t \ge 0.
\end{equation} 

With this choice of $\alpha^i$, $i=1,2,3$, $Z$ in \eqref{eq:densityQ} is a $\mathbb{G}$-adapted process such that $[M,Z]=0$ a.s.. Applying Corollary VIII.1.16 of \citep{RevuzYor}, it can be seen that for these choices $Z$ is also a true martingale. For this reason, defining $Q$ by
$$
\frac{dQ}{dP}|_{\G_t} = Z_t, \quad t \ge 0,
$$
we have that $Q \in \M_{loc}(M, \bG)$, $Q \ne P$.

However, Theorem \ref{thm:invariance} implies that $\M_{M}(M,\bG) = \emptyset$ for every filtration $\bG$ to which $M$ is adapted, i.e., there does not exist any measure $Q \sim P$ such that $M$ is a true martingale under $Q$. This also means that, in the following analysis, property \eqref{prob:locmart} holds if and only if \eqref{prob:intersect} holds.
\end{remark}

\subsection{Optional projection into the filtration generated by $B^1$ and $B^2$}\label{subsec:B1B2}

We let $\bG$ be the natural filtration of $B^1$, $B^2$ and $B^3$, and $\bF$ be generated by $B^1$ and $B^2$. We still denote the optional projection of the inverse three-dimensional Bessel process $M$ into $\bF$ by $\op M$. Theorem 5.2 of \citep{follmer2011local} states that $\op M$ is an $\bF$-local martingale and has the form
$$
\op M_t = u(B^1_t,B_t^2,t),
$$
with
$$
u(x,y,t)=\frac{1}{\sqrt{2\pi t}}\exp\left(\frac{x^2+y^2}{4t}\right)K_0\left(\frac{x^2+y^2}{4t}\right),
$$
where we denote by $K_n$, $n \ge 1$, the modified Bessel functions of the second kind. In particular, it holds
\begin{equation}\label{eq:deru}
\partial_xu(x,y,t)=x \psi(x,y,t), \qquad \partial_yu(x,y,t)=y \psi(x,y,t),
\end{equation}
where
\begin{equation}\label{eq:psimodbessel}
\psi(x,y,t)=\frac{1}{\sqrt{2\pi t}}\exp\left(\frac{x^2+y^2}{4t}\right)\left(K_0\left(\frac{x^2+y^2}{4t}\right)-K_1\left(\frac{x^2+y^2}{4t}\right)\right).
\end{equation}
Since $\op M$ is an $\bF$-local martingale, we focus here on property \eqref{prob:equal}. We start by the following 

\begin{lemma}\label{lem:sameproj}
Let $Q$ be a probability measure equivalent to $P$, such that  $\oG Z$ is $\bF$-adapted. Then it holds $\opQ M = \op M$.
\end{lemma}
\begin{proof}
Define $\otF Z$ by $\otF Z_t:=\bE[dQ/dP|\F_t]$.  
We have that
$$
^{\G}Z_t =
 \otF Z_t, \quad t \ge 0,
$$
and then it holds
\begin{align}
\opQ M_t = \mathbb{E}^Q[M_t| \F_t]=(^{ \F}Z_t)^{-1}\mathbb{E}[^{\G}Z_tM_t| \F_t]=\mathbb{E}[M_t| \F_t]= \op M_t, \quad t \ge 0.\notag
\end{align}
\end{proof}

We can now give the following theorem, which provides a positive answer to property \eqref{prob:equal} in this example.

\begin{theorem}\label{thm:invbessel12}
Let $\bF$ be the natural filtration of $B^1$ and $B^2$. Then it holds 
$$
\M_{loc}(M, \bG, \bF)=\M_{loc}^o(M, \bG, \bF).
$$
\end{theorem}

\begin{proof}
We first prove that $\M_{loc}(M, \bG, \bF) \subseteq \M^o_{loc}(M, \bG, \bF)$. 

Introduce the sequence of stopping times $(\tau_n)_{n \in \bN}$ with
$$
\tau_n = \inf\left\{(B_t^1)^2+(B_t^2)^2 \le \frac{1}{n}\right\}, \quad n \ge 1.
$$
Since $\lim_{n \to \infty} \tau_n = \infty$ because the origin $(0,0)$ is polar for a two-dimensional Brownian motion, this is a localizing sequence of $\bF$-stopping times that makes $M$ a bounded martingale. Consider now $Q \in \M_{loc}(M, \bG,  \bF)$, i.e., suppose that $M$ is a $(Q,\bG)$-local martingale, and call $\opQ M$ the $Q$-optional projection of $M$ into $ \bF$. Theorem \ref{thm:samesequence}  implies that  $\opQ M$ is a $(Q, \bF)$-local martingale, i.e., $Q \in \M^o_{loc}(M, \bG,  \bF)$.\\
We now prove that $\M_{loc}^o(M, \bG, \bF) \subseteq \M_{loc}(M, \bG, \bF)$. Take $Q \in \M^o_{loc}(M, \bG, \bF)$, i.e., suppose that $\opQ M$ is a $(Q, \bF)$-local martingale. \\
Since $\op M$ is a $(P,\bF)$-local martingale and $\opQ M = \op M$ by Lemma \ref{lem:sameproj}, from $Q \in \M_{loc}^o(M, \bG,  \bF)$ it follows that $[ \oG Z,\op M]$ is a local martingale,
because $\oG Z=  \otF Z$ by the proof of Lemma \ref{lem:sameproj}.\\
Note now that since the density of $Q$ with respect to $P$ is $ \bF$-adapted, it holds 
\begin{equation}\label{eq:densityQ2}
\oG Z_t=\frac{dQ}{dP}|_{\G_t}=\mathcal{E}\left(\int_0^{t} \alpha_s^1dB_s^1+\int_0^{t} \alpha_s^2dB_s^2\right), \quad t \ge 0,
\end{equation}
where $\alpha^1$ and $\alpha^2$ are $\bF$-adapted processes, and the  Dol\'eans exponential in \eqref{eq:densityQ2} is well defined and a true martingale.

By \eqref{eq:deru} and \eqref{eq:densityQ2}, we have 
\begin{equation}\label{eq:bracketZoM}
 [\oG Z,\op M]_t = \int_0^t \oG Z_s \psi(B^1_s,B_s^2,s) (\alpha_s^1 B^1_s+\alpha^2_sB_s^2)ds, \quad t \ge 0,
\end{equation}
where $\psi$ is defined in \eqref{eq:psimodbessel}. Since $\psi(x,y,t)<0$ for $x,y<\infty$ and $t>0$, see for example \citep{yang2017approximating}, equation \eqref{eq:bracketZoM} together with the fact that $[ \oG Z,\op M]$ is a local martingale implies that 
\begin{equation}\label{eq:alphaB}
\alpha_t^1 B^1_t+\alpha^2_tB_t^2=0,  \quad P\text{-a.s.}, \quad t \ge 0,
\end{equation}
 as $P$ is equivalent to $Q$. Moreover, from \eqref{Mdynamics} and \eqref{eq:densityQ2} it follows that 
$$
[\oG Z,M]_t = - \int_0^t Z_sM_s^3(\alpha_s^1 B^1_s+\alpha^2_sB_s^2)ds, \quad t \ge 0,
$$
and this is zero $P$-a.s. by \eqref{eq:alphaB}. Since $M$ is $(P,\bG)$-local martingale, this implies that $M$ is also a $(Q,\bG)$-local martingale. Hence $Q \in \M_{loc}(M, \bG, \bF)$.

\end{proof}

\subsection{Delayed information}\label{subsec:delayed}
We now consider a market model with delayed information: here $\bG$ is the filtration generated by the Brownian motion $W$ in \eqref{eq:W}, whereas $\bF=(\F_t)_{t \ge 0}$ is given by $\F_t=\G_{t-\epsilon}$, $\epsilon>0$. As explained above, this means that investors have access to the information about $W$, with respect to which $M$ is adapted by \eqref{eq:Mnewdynamics}, only with a positive delay $\epsilon$. We start our analysis with the following

\begin{lemma}\label{lemma:delayedM}
For every $\epsilon>0$, it holds
$$
\bE[M_{t+\epsilon}|\sigma(B^1_t,B_t^2,B_t^3)]=M_t \erf\left(\frac{1}{M_t \sqrt{2\epsilon}}\right),
$$
where $\erf(x):=\frac{2}{\sqrt{\pi}}\int_{0}^x e^{-t^2}dt$.
\end{lemma}
\begin{proof}
We have
$$
\bE[M_{t+\epsilon}|\sigma(B^1_t,B_t^2,B_t^3)]=u(\epsilon, B^1_t, B^2_t, B_t^3),
$$
with 
$$
u(t,a,b,c)=(2 \pi t)^{-3/2}\int_{-\infty}^{\infty}\int_{-\infty}^{\infty}\int_{-\infty}^{\infty}\frac{e^{-\frac{1}{2t}\left((x-a)^2+(y-b)^2+(z-c)^2\right)}}{\sqrt{x^2+y^2+z^2}}dzdydx =: \rm{I}.
$$
We set $R=\sqrt{a^2+b^2+c^2}$, $r=\sqrt{x^2+y^2+z^2}$. Applying a suitable change of variables, the above integral can be written in spherical coordinates as
\begin{align}
\rm{I} &= (2 \pi t)^{-3/2} \int_0^{2\pi}\int_0^{\infty}r^2\frac{1}{r}\int_0^{\pi} \sin(\theta)e^{-\frac{1}{2t}(r^2-2rR\cos(\theta)+R^2)}d\theta drd\phi\notag \\
&=\frac{2}{R\sqrt{\pi}} \int_0^{\frac{R}{\sqrt{2t}}} e^{-r^2}dr  = \frac{1}{\sqrt{a^2+b^2+c^2}} \erf\left( \sqrt{\frac{a^2+b^2+c^2}{2t}}\right).\notag
\end{align}
Thus 
\begin{align}
\bE[M_{t+\epsilon}|\sigma(B^1_t,B_t^2,B_t^3)]&= \frac{1}{\sqrt{(B_t^1)^2+(B_t^2)^2+(B_t^3)^2}} \erf\left( \sqrt{\frac{(B_t^1)^2+(B_t^2)^2+(B_t^3)^2}{2\epsilon}}\right) \notag \\
&=M_t \erf\left(\frac{1}{M_t \sqrt{2\epsilon}}\right).\notag
\end{align}
\end{proof}

\begin{proposition}\label{proposition:delayed}
Let $\bG=(\G_t)_{t \ge 0}$ be the filtration generated by the Brownian motion $W$ in \eqref{eq:W}, and $\bF=(\F_t)_{t \ge 0}$ be given by $\F_t:=\G_{t -\epsilon}$, $t \ge 0$, $\epsilon>0$. Thus
$$
\op M_{t+\epsilon} = \bE[M_{t+\epsilon}|\F_{t+\epsilon}] = M_{t} \erf\left(\frac{1}{M_{t} \sqrt{2\epsilon}}\right), \quad t \ge 0.
$$
\end{proposition}
\begin{proof}
Due to the Markov property of $W$ and to the fact that $\sigma(W_t) \subset \sigma(B^1_t,B_t^2,B_t^3)$ by \eqref{eq:Mnewdynamics} and \eqref{eq:W}, from Lemma \ref{lemma:delayedM} it follows 
\begin{align}
\bE[M_{t+\epsilon}|\F_{t+\epsilon}]&=\bE[M_{t+\epsilon}|\sigma(W_t)]=\bE\left[\bE\left[M_{t+\epsilon}|\sigma(B^1_t,B_t^2,B_t^3)\right]|\sigma(W_t)\right] \notag \\
&= \bE\left[M_t \erf\left(\frac{1}{M_t \sqrt{2\epsilon}}\right)\bigg|\sigma(W_t)\right] = M_t \erf\left(\frac{1}{M_t \sqrt{2\epsilon}}\right),\quad t \ge 0,\notag
\end{align} 
as $M_t$ is $\sigma(W_t)$-measurable.
\end{proof}
By Proposition \ref{proposition:delayed}, we have that 
$$
\op M_{t+\epsilon}=f(M_t), \quad t \ge 0,
$$
with $f(x)=x\cdot \erf\left(\frac{1}{x\sqrt{2\epsilon}}\right)$. Since
$$
f'(x)=-\frac{\sqrt{2}e^{-\frac{1}{2 \epsilon x^2}}}{x\sqrt{\pi\epsilon} }+\erf\left(\frac{1}{x\sqrt{2 \epsilon}}\right), \qquad f''(x)=-\frac{\sqrt{2}\epsilon^{-\frac{3}{2}}e^{-\frac{1}{2 \epsilon x^2}}}{x^4\sqrt{\pi} } ,
$$
applying It\^o's formula we obtain
{\small{
\begin{align}
d\op M_{t + \epsilon} &= \left(-\frac{\sqrt{2}e^{-\frac{1}{2 \epsilon M_t^2}}}{M_t\sqrt{\pi\epsilon} }+\erf\left(\frac{1}{M_t\sqrt{2 \epsilon}}\right)\right)dM_t-\frac{\sqrt{2}\epsilon^{-\frac{3}{2}}e^{-\frac{1}{2 \epsilon M_t^2}}}{M_t^4\sqrt{\pi} }d[M,M]_t  \notag \\ &=
\left(\frac{\sqrt{2}e^{-\frac{1}{2 \epsilon M_t^2}}}{\sqrt{\pi\epsilon} }M_t-\erf\left(\frac{1}{M_t\sqrt{2 \epsilon}}\right)M_t^2\right)dW_t-\sqrt{\frac{2}{\pi}}\epsilon^{-\frac{3}{2}}e^{-\frac{1}{2 \epsilon M_t^2}}dt, \label{eq:itoM}
\end{align}}}
By the above expression, we note that the optional projection is a strict $\bF$- supermartingale, as the drift is strictly negative. Since by Remark \ref{remark:completeincomplete}
we have $\M_{loc}(M,\bG) = \{P\}$, this implies that 
\begin{equation}\label{eq:intersectionemptyset}
\M_{loc}(M,\bG) \cap \M^o_{loc}(M, \bF) = \emptyset,
\end{equation}
i.e., properties \eqref{prob:locmart}, \eqref{prob:intersect} and \eqref{prob:equal} do not hold.

Moreover, we give the following theorem, which implies that properties \eqref{prob:locmartgen} and \eqref{prob:union} are not satisfied.
\begin{theorem}\label{thm:Mlocempty}
Let $\bG=(\G_t)_{t \ge 0}$ be the filtration generated by the Brownian motion $W$ in \eqref{eq:W}, and $\bF=(\F_t)_{t \ge 0}$, with $\F_t:=\G_{t -\epsilon}$, $t \ge 0$, $\epsilon>0$. Thus
$$
\M_{loc}(\op M,\bF) = \emptyset.
$$
\end{theorem}
To prove Theorem \ref{thm:Mlocempty}, we rely on some results provided by \citep{MijatovicUrusov}, which we now recall.
Consider the state space $J=(l,r)$, $-\infty\le l < r \le \infty$ and a $J$-valued diffusion $Y=(Y_t)_{t \ge 0}$ on some filtered probability space, governed by the SDE
\begin{equation}\label{ymij}
dY_t=\mu_Y(Y_t)dt+\sigma_Y(Y_t)dB_t, \quad t \ge 0,
\end{equation}
with $Y_0=x_0 \in J$, $B$ Brownian motion and deterministic functions $\mu_Y(\cdot)$ and $\sigma_Y(\cdot)$, that from now on we will simply denote by $\mu_Y$ and $\sigma_Y$, such that 
\begin{equation}\label{condsigma1}
\sigma_Y(x) \ne 0 \quad \forall x \in J
\end{equation}
 and 
 \begin{equation}\label{condsigma2}
 \frac{1}{\sigma_Y^2},\text{ }  \frac{\mu_Y}{\sigma_Y^2} \in L_{loc}^1(J),
 \end{equation}
 where $ L_{loc}^1(J)$ denotes the class of locally integrable functions $\psi$ on $J$, i.e., the measurable functions $\psi: (J,\mathcal{B}(J)) \rightarrow (\mathbb{R},\mathcal{B}(\mathbb{R}))$ that are integrable on compact subsets of $J$. \\
Consider the stochastic exponential
\begin{equation}\label{stocexp}
\E\left(\int_0^t g(Y_u)dB_u\right), \quad t \ge 0,
\end{equation}
with $g(\cdot)$ such that 
\begin{equation}\label{condb1}
\frac{g^2}{\sigma_Y^2} \in L_{loc}^1(J).
\end{equation}
 Put $\bar{J}=[l,r]$ and, fixing an arbitrary $c \in J$, define
\begin{align}
&\rho(x):=\exp\left\{-\int_c^x\frac{2\mu_Y}{\sigma_Y^2}(y)dy\right\}, \quad x \in J, \label{rho} \\
&\tilde{\rho}(x):=\rho(x)\exp\left\{-\int_c^x\frac{2g}{\sigma_Y}(y)dy\right\}, \quad x \in J, \label{rhotilde} \\
&s(x):=\int_c^x\rho(y)dy, \quad x \in \bar{J}, \label{s} \\
& \tilde{s}(x):=\int_c^x\tilde{\rho}(y)dy,\quad x \in \bar{J}. \label{stilde}
\end{align}
Denote $\rho=\rho(\cdot)$, $s=s(\cdot)$, $s(r)=\lim_{x \to r^-}s(x)$,  $s(l)=\lim_{x \to l^+}s(x)$, and analogously for $\tilde{s}(\cdot)$ and $\tilde{\rho}(\cdot)$.\\
Define 
$$
L_{loc}^1(r-):=\{\psi:(J,\mathcal{B}(J)) \rightarrow (\mathbb{R},\mathcal{B}(\mathbb{R})) \Big| \int_x^r |\psi(y)|dy<\infty \text{ for some } x \in J\},
$$ 
and   $L_{loc}^1(l+)$ analogously.
We report here Theorem 2.1 in \citep{MijatovicUrusov}.

\begin{theorem}\label{thm:uru}
Let the functions $\mu_Y$, $\sigma_Y$, and $g$ satisfy conditions \eqref{condsigma1}, \eqref{condsigma2} and  \eqref{condb1}, and let $Y$ be a solution of the SDE \eqref{ymij}. 
\\ Then the Dol\'eans exponential given by \eqref{stocexp} is a true martingale if and only if both of the following requirements are satisfied:
\begin{enumerate}
\item it does not hold
\begin{equation}\label{cond1}
\tilde{s}(r)<\infty \quad \text{and} \quad \frac{\tilde{s}(r)-\tilde{s}}{\tilde{\rho}\sigma_Y^2} \in L_{loc}^1(r-),
\end{equation}
or it holds
\begin{equation}\label{cond3}
s(r)<\infty \quad \text{and} \quad \frac{(s(r)-s)g^2}{\rho\sigma_Y^2} \in L_{loc}^1(r-);
\end{equation}
\item it does not hold
\begin{equation}\notag
\tilde{s}(l)>-\infty \quad \text{and} \quad \frac{\tilde{s}-\tilde{s}(l)}{\tilde{\rho}\sigma_Y^2} \in L_{loc}^1(l+),
\end{equation}
or it holds
\begin{equation}\notag
s(l)>-\infty \quad \text{and} \quad \frac{(s-s(l))g^2}{\rho\sigma_Y^2} \in L_{loc}^1(l+).
\end{equation}
\end{enumerate}  
\end{theorem}
We now use Theorem \ref{thm:uru} in order to prove Theorem \ref{thm:Mlocempty}.

\begin{proof}[Proof of Theorem \ref{thm:Mlocempty}]
By equation \eqref{eq:itoM} we have that 
$$
d\op M_{t + \epsilon}=\mu(M_t)dt+\sigma(M_t)dW_t, \quad t \ge 0,
$$
with 
\begin{equation}\label{eq:muandsigma}
\mu(x)=-\sqrt{\frac{2}{\pi}}\epsilon^{-\frac{3}{2}}e^{-\frac{1}{2 \epsilon x^2}}, \qquad \sigma(x)=x\sqrt{\frac{2}{\pi\epsilon}}e^{-\frac{1}{2 \epsilon x^2}}-x^2\erf\left(\frac{1}{x\sqrt{2 \epsilon}}\right).
\end{equation}
By Girsanov's Theorem there exists a probability measure $Q \in \M_{loc}(\op M,\bF)$ if the Dol\'eans exponential
\begin{equation}\label{eq:Zgirsanov}
\frac{dQ}{dP}|_{\G_t}=Z_t=\mathcal{E}\left(\int_0^{t} \alpha_sdW_s\right), \quad t \ge 0,
\end{equation}
with 
\begin{equation}\label{eq:alphaM}
\alpha_t=-\frac{\mu(M_t)}{\sigma(M_t)}, \quad t \ge 0.
\end{equation}
 is a true martingale. 

In order to prove that this does not hold, we apply Theorem \ref{thm:uru}. In our case, by equations \eqref{eq:Mnewdynamics}, \eqref{eq:muandsigma} and \eqref{eq:alphaM}, we have $Y=M$, $J=(0,\infty)$, $\mu_Y \equiv 0$, $\sigma_Y(x)=-x^2$ and
$$
g(x)=\frac{\sqrt{2/\pi}\epsilon^{-\frac{3}{2}}e^{-\frac{1}{2 \epsilon x^2}}}{x\left(\sqrt{\frac{2}{\pi\epsilon}}e^{-\frac{1}{2 \epsilon x^2}}-x\cdot\erf\left(\frac{1}{x\sqrt{2 \epsilon}}\right)\right)}.
$$
Note that condition \eqref{condsigma1} and \eqref{condsigma2} are satisfied. In order to prove that \eqref{condb1} also holds, it is enough to check that 
$$
x\cdot \erf\left(\frac{1}{x\sqrt{2 \epsilon}}\right) - \sqrt{\frac{2}{\pi\epsilon}}e^{-\frac{1}{2 \epsilon x^2}} > 0 \quad \text{for every }x \in (0,\infty).
$$
This is true if and only if 
$$
\erf(y)\frac{1}{y\sqrt{2\epsilon}} - \sqrt{\frac{2}{\pi\epsilon}}e^{-y^2} > 0  \quad \text{for every }y \in (0,\infty),
$$
i.e., if and only if 
$$
F(y):=\erf(y) - \frac{2}{\sqrt{\pi}}ye^{-y^2} > 0 \quad \text{for every }y \in (0,\infty).
$$
The last condition holds, since $F(0)=0$ and $F'(y) =  \frac{4}{\sqrt{\pi}}y^2e^{-y^2}>0$ for every $y>0$, and the assumptions of Theorem \ref{thm:uru} are thus satisfied.

We now show that condition \eqref{cond3} fails whereas \eqref{cond1} is satisfied, implying that the density $Z$ introduced in \eqref{eq:Zgirsanov} is not a martingale. 

Consider first $\rho$ and $s$ defined in \eqref{rho} and \eqref{s}, respectively. We have $\rho \equiv 1$,  so that $s(x)=x-c$, for any $c>0$. This implies that $s(\infty)=+\infty$, so that condition \eqref{cond3} fails.

We now check condition \eqref{cond1}. We have that
$$
\lim_{x \to \infty}\frac{e^{-\frac{1}{2 \epsilon x^2}}}{x^2\left(\sqrt{\frac{2}{\pi\epsilon}}e^{-\frac{1}{2 \epsilon x^2}} - x\cdot\erf\left(\frac{1}{x\sqrt{2 \epsilon}}\right)\right)}=-3\sqrt{\frac{\pi}{2}}\epsilon^{3/2},
$$
so that 
$$
\lim_{x \to \infty} -x\frac{2g(x)}{\sigma_Y(x)}=\lim_{x \to \infty} 2\sqrt{2/\pi}\epsilon^{-\frac{3}{2}}\frac{e^{-\frac{1}{2 \epsilon x^2}}}{x^2\left(\sqrt{\frac{2}{\pi\epsilon}}e^{-\frac{1}{2 \epsilon x^2}}-x\cdot\erf\left(\frac{1}{x\sqrt{2 \epsilon}}\right)\right)}=-6.
$$
Hence we have that for every $\delta>0$, there exists $\bar{x}>0$ such that 
\begin{equation}\label{eq:deltalimit}
\big|-x\frac{2g(x)}{\sigma_Y(x)}+6\big| \le \delta \quad \text{for every }x \ge \bar{x}.
\end{equation}
 We fix $\delta<1$ and choose $\bar x>0$ such that \eqref{eq:deltalimit} holds.
For every $x>\bar x$ it holds
\begin{align}
&\bigg|-\int_{\bar x}^x\frac{2g(y)}{\sigma_Y(y)}dy+\int_{\bar x}^x \frac{6}{y}dy\bigg| \le  \int_{\bar x}^x\bigg|-\frac{2g(y)}{\sigma_Y(y)}+ \frac{6}{y}\bigg|dy \le \int_{\bar x}^x\frac{1}{y}\bigg|-y\frac{2g(y)}{\sigma_Y(y)}+ 6\bigg|dy  \notag \\ & \le \delta \left(\log(x) - \log(\bar x)\right). \notag
\end{align}
Thus for every $x>\bar x$ we have 
$$
(-6-\delta) \left(\log(x) - \log(\bar x)\right) \le-\int_{\bar x}^x\frac{2g(y)}{\sigma_Y(y)}dy \le (-6+\delta) \left(\log(x) - \log(\bar x)\right).
$$
Therefore, taking $\tilde{\rho}$ as in \eqref{rhotilde} and choosing $c=\bar x$, for every $x > \bar x$ it holds
\begin{equation}\label{eq:inequalitiesrhotilde}
\left(\frac{x}{\bar{x}}\right)^{-6-\delta} \le \tilde \rho(x) \le \left(\frac{x}{\bar{x}}\right)^{-6+\delta}.
\end{equation}
Hence, taking $\tilde{s}$ as in \eqref{stilde} and choosing again $c = \bar x$, for every $x > \bar x$ it holds
$$
\int_{\bar x}^x \left(\frac{y}{\bar{x}}\right)^{-6-\delta} dy \le \tilde{s}(x)=\int_{\bar x}^x\tilde{\rho}(y)dy \le \int_{\bar x}^x \left(\frac{y}{\bar{x}}\right)^{-6+\delta} dy,
$$
so that $\tilde{s}(\infty)<\infty$ and in particular
$$
\bar x^{6+\delta}\frac{x^{-5-\delta}}{5+\delta} \le \tilde s(\infty)-\tilde s(x) = \int_{x}^{\infty}\tilde{\rho}(y)dy \le\bar x^{6-\delta}\frac{x^{-5+\delta}}{5-\delta}.
$$
Together with \eqref{eq:inequalitiesrhotilde}, this implies that for every $x > \bar x$ it holds
$$
\frac{\bar x^{2 \delta}}{5+\delta} x^{1-2\delta} \le \frac{\tilde{s}(\infty)-\tilde{s}(x)}{\tilde{\rho}(x)} \le  \frac{\bar{x}^{-2\delta}}{5-\delta} x^{1+2\delta},
$$
and thus
$$
\frac{\bar x^{2 \delta}}{5+\delta} x^{-3-2\delta} \le \frac{\tilde{s}(\infty)-\tilde{s}(x)}{\sigma_Y^2(x)\tilde{\rho}(x)} \le  \frac{\bar{x}^{-2\delta}}{5-\delta} x^{-3+2\delta}.
$$
Therefore, as $\delta<1$ by the choice of $\bar x$, we have that $\frac{\tilde{s}(\infty)-\tilde{s}}{\tilde{\rho}\sigma_Y^2} \in L_{loc}^1(\infty-)$ and condition \eqref{cond1} holds. By Theorem \ref{thm:uru}, it follows that $Z$ defined in \eqref{eq:Zgirsanov} is not a martingale.  
\end{proof}

We give now an example of a process whose optional projection into the delayed filtration is not a local martingale but admits an equivalent local martingale measure.

\begin{example}\label{ex:invbesslminusbm}
Consider again the filtration $\bG=(\G_t)_{t \ge 0}$ generated by the Brownian motion $W$ in \eqref{eq:W}, and define $\bF=(\F_t)_{t \ge 0}$, with $\F_t:=\G_{t -\epsilon}$, $t \ge 0$, $\epsilon>0$. 
Introduce the process $X = M - W,$ where $M$ is as usual the inverse three-dimensional Bessel process. 
Thus
$$
\op X_{t+\epsilon} = \bE[X_{t+\epsilon}|\G_{t}] = \bE[M_{t+\epsilon} - W_{t+\epsilon}|\G_{t}] =M_{t} \erf\left(\frac{1}{M_{t} \sqrt{2\epsilon}}\right) - W_t, \quad t \ge 0,
$$
where the last inequality comes from Proposition \ref{proposition:delayed} and from the martingale property of $W$. From \eqref{eq:itoM} it holds therefore
$$
d \op X_{t+\epsilon} = \left(\sqrt{\frac{2}{\pi\epsilon}}e^{-\frac{1}{2 \epsilon M_t^2}}M_t-\erf\left(\frac{1}{M_t\sqrt{2 \epsilon}}\right)M_t^2 -1\right)dW_t-\sqrt{\frac{2}{\pi}}\epsilon^{-\frac{3}{2}}e^{-\frac{1}{2 \epsilon M_t^2}}dt, \quad t \ge 0.
$$
It is then clear that $\op X$ is not an $\bF$-local martingale. This implies that 
\begin{equation}\notag
\M_{loc}(X,\bG) \cap \M^o_{loc}(X, \bF) = \emptyset,
\end{equation}
since $\M_{loc}(X,\bG) = \{P\}$ by Remark \ref{remark:completeincomplete}, so that properties \eqref{prob:locmart}, \eqref{prob:intersect} and \eqref{prob:equal} have a negative answer.

We now introduce the Dol\'eans exponential
\begin{equation}\notag
\bar Z_t=\mathcal{E}\left(\int_0^{t} \bar \alpha_sdW_s\right), \quad t \ge 0,
\end{equation}
with 
\begin{equation}\notag
\bar \alpha_t=\frac{\sqrt{\frac{2}{\pi}}\epsilon^{-\frac{3}{2}}e^{-\frac{1}{2 \epsilon M_t^2}}}{\sqrt{\frac{2}{\pi\epsilon}}e^{-\frac{1}{2 \epsilon M_t^2}}M_t-\erf\left(\frac{1}{M_t\sqrt{2 \epsilon}}\right)M_t^2 -1}, \quad t \ge 0,
\end{equation}
and define the measure $\bar Q$ by $\frac{d\bar Q}{dP}|_{\G_t}=\bar Z_t$, $t \ge 0$. Since 
$$
\erf\left(\frac{1}{x\sqrt{2 \epsilon}}\right)x - \sqrt{\frac{2}{\pi\epsilon}}e^{-\frac{1}{2 \epsilon x^2}} > 0 \quad \text{for every }x \in (0,\infty),
$$ 
as we have shown in the proof of Theorem \ref{thm:Mlocempty}, it holds $|\bar \alpha_t| \le \sqrt{\frac{2}{\pi}}\epsilon^{-\frac{3}{2}}$ for all $t \ge 0$. Thus Corollary \rm{VIII}.1.16 of \citep{RevuzYor}\footnote{If $L$ is a local martingale such that either $\exp(\frac{1}{2}L)$ is a submartingale or $$\bE\left[\exp\left(\frac{1}{2}\langle L,L \rangle_t\right)\right]<\infty$$ for every $t$, then $\E(L)$ is a martingale.} and  Girsanov's Theorem imply that  $\bar Q \in \M_{loc}(\op X,\bG)$. Hence, properties \eqref{prob:locmartgen} and \eqref{prob:union} are satisfied.
\end{example}

\begin{remark}\label{rem:delayedextended}
In the above analysis $\bG$ represents the natural filtration of $M$. By \eqref{eq:intersectionemptyset} we obtain that \eqref{prob:locmart} is not satisfied in this setting. However, \eqref{prob:locmart} still does not hold if $\bG$ is given by the filtration generated by $(B^1,B^2,B^3)$. In this case there exist infinitely many measures in $\M_{loc}(M,\bG)$. By Corollary \ref{cor:sameprojection} we obtain that for every $Q \in \M_{loc}(M,\bG)$ it holds 
$$
\bE^Q[M_t|\F_t]=\bE^P[M_t|\F_t], \quad a.s., \quad t \ge 0.
$$ 
Then $\opQ M$ is a $(Q,\bF)$-local martingale if and only if $Q$ is an equivalent local martingale measure for $\op M$, which has dynamics given in \eqref{eq:itoM}. However, this cannot be the case because such a measure $Q$ would be defined by a density which is not a true martingale, by the same arguments as in the proof of Theorem \ref{thm:Mlocempty}.

%
%

\end{remark}

\section{A stochastic volatility example}\label{sec:stocvol}

Introduce a three-dimensional Brownian motion $B = (B^1,B^2,B^3)$ on a filtered probability space $(\Omega, \mathcal{F}, P, \bG = (\mathcal{G}_t)_{t \ge 0})$, and consider a stochastic volatility model of the form 
\begin{align}
&dS_t=\sigma_1v_t^{\alpha}S_tdB^1_t + \sigma_2v^{\alpha}_tS_tdB_t^2, \quad t \ge 0, \quad S_0 = s>0, \label{eq:S} \\
&dv_t=a_1v_tdB^1_t + a_2v_tdB_t^2 + a_3v_tdB_t^3 + \rho(L-v_t)dt, \quad t \ge 0, \quad v_0 = 1,\label{eq:v}
\end{align}
where $\alpha, \rho, L \in \reals^+$ and $\sigma_1, \sigma_2, a_1, a_2, a_3 \in \reals$.

\begin{remark}
The class of stochastic volatility processes \eqref{eq:S}-\eqref{eq:v} reduces to the class considered in \citep{sin1998complications} when $a_3=0$ and to the class presented in \citep{Biagini} when $\rho=0$ and $\alpha=1$. Therefore, all the results of this section can be applied to these particular cases.
\end{remark}
The next proposition states that, under a given condition on the coefficients of \eqref{eq:S}-\eqref{eq:v}, $S$ is a strict $\bG$-local martingale under $P$ but $\M_M(S, \bG) \ne \emptyset$. 
\begin{proposition}\label{prop:sin}
Consider the unique strong solution\footnote{Existence and uniqueness of a strong solution to \eqref{eq:S}-\eqref{eq:v} can be proved as an extension of \citep[Remark 2.2]{sin1998complications}.} $(S,v)$ to the system of SDEs \eqref{eq:S}-\eqref{eq:v}. Then:
\begin{enumerate}
\item $S$ is a local martingale, and is a true martingale if and only if $$a_1 \sigma_1 + a_2 \sigma_2 \le 0.$$
\item For every $T>0$ there exists a probability measure $Q$ equivalent to $P$ on $\G_T$ such that $S$ is a true $Q$-martingale on $[0,T]$.
\end{enumerate}
\end{proposition}
\begin{proof} The proofs of the two claims are easy extensions of the proofs of Theorem 3.2 of \citep{sin1998complications} and of Theorem 5.1 of \citep{Biagini}, respectively.
\end{proof}

We now give two results that provide a relation between the expectation of $S$ and the explosion time of a process associated to the volatility $v$. 
\begin{lemma}\label{lemma:sin1}
 Let $(S,v)$ satisfy the system of SDEs \eqref{eq:S}-\eqref{eq:v}. Then
\begin{equation}\label{eq:mt}
\bE[S_t]=S_0P(\{\hat v \text{ does not explode on }[0,t]\}), \quad t \ge 0,
\end{equation}
where $\hat v=(\hat v_t)_{t \ge 0}$ is given by 
\begin{align}\label{eq:w}
d\hat v_t=&a_1\hat v_tdB^1_t + a_2\hat v_tdB_t^2 +a_3\hat v_tdB^3_t+\rho(L-\hat v_t)dt \notag \\
&+ (a_1 \sigma_1 + a_2 \sigma_2)\hat v^{\alpha+1}_t dt, \quad t \ge 0,
\end{align}
$\hat v_0=1$.
\end{lemma}
\begin{proof} This result is a particular case of Proposition \ref{prop:stopped}, which we give below. 
\end{proof}
\begin{lemma}\label{lemma:sin2}
The (unique) solution to equation \eqref{eq:w}
explodes to $+\infty$ in finite time with positive probability if and only if $a_1\sigma_1+a_2\sigma_2>0$. Moreover, if $a_1\sigma_1+a_2\sigma_2>0$, it does not reach zero in finite time.
\end{lemma}
\begin{proof} 
The result is given in Lemma 4.3 of \citep{sin1998complications} when $a_3=0$, and proved by using Feller's test of explosions. In this case, the test is applicable because $\hat v$ is a one-dimensional It\^o diffusion with respect to the Brownian motion $1/|a|(a \cdot B)$, with $a=(a_1,a_2)$ and $B=(B^1,B^2)$. The author proves that $\hat v$ explodes with positive probability in finite time, and does not reach the origin in finite time, when $a \cdot \sigma >0$, where $\sigma = (\sigma_1,\sigma_2)$. In our case, the proof comes as an easy extension by considering now $a=(a_1,a_2,a_3)$ and $\sigma = (\sigma_1,\sigma_2,0)$.
\end{proof}


We now give an example for which the property \eqref{prob:locmartgen} is satisfied. We start with the following lemma, which is Proposition 5.2 of \citep{karatzas2016distribution}. 

\begin{lemma}\label{lem:explosionC1}
Fix an open interval $I=(\ell,r)$ with $-\infty \le \ell <r \le \infty$ and consider the stochastic differential equation 
\begin{equation}\label{eq:karatzas}
dY_t=s(Y_t)\left(dW_t+b(Y_t)dt\right), \quad t \ge 0, \quad Y_0=\xi,
\end{equation}
where $\xi \in I$ and $W$ denotes a Brownian motion. Suppose that the functions $b: (I,\B(I)) \to (\reals,\B(\reals))$ and $s: (I,\B(I)) \to (\reals \setminus \{0\},\B(\reals \setminus \{0\}))$ are measurable and satisfy
\begin{equation}\label{eq:condkaratzas}
\int_K \left( \frac{1}{s^2(y)}+\left|\frac{b(y)}{s(y)}\right|\right)dy < \infty \qquad \text{for every compact set $K\subset I$}.
\end{equation}
Call $\tau^{\xi}$ the first time when the (unique in the sense of probability distribution) weak solution $Y$ to \eqref{eq:karatzas} exits the open interval $I$, and introduce the function $U:(0,\infty) \times I \to \reals^+$ defined by
$$
U(t,\xi):=P(\tau^{\xi}>t).
$$ 
If the functions $s(\cdot)$ and $b(\cdot)$ are locally H\"older continuous on $I$, the function $U(\cdot,\cdot)$ is of class $C\left([0,\infty)\times I\right) \cap C^{1,2}\left((0,\infty)\times I\right)$.
\end{lemma}
Applying Lemma \ref{lem:explosionC1} to our setting, we get the following result. 

\begin{lemma}\label{lem:mC1}
Consider the solution $\hat v$ to equation \eqref{eq:w}, supposing $\rho = 0$ and $a_1 \sigma_1 + a_2 \sigma_2>0$. Define the function $m:(0,\infty) \to \reals^+$ by 
$$
m(t) = P(\{\hat v \text{ does not explode on }[0,t]\}).
$$
Thus $m \in C^1((0,\infty))$.
\end{lemma}
\begin{proof}
Note that $\hat v$ is a one-dimensional It\^o diffusion with respect to the Brownian motion $W=1/|a|(a \cdot B)$, with $a=(a_1,a_2,a_3)$ and $B=(B^1,B^2,B^3)$. In particular, it holds
$$
d\hat v_t = |a|\hat v_t dW_t + (a_1 \sigma_1 + a_2 \sigma_2)\hat v ^{\alpha +1}dt, \quad t \ge 0.
$$
We are thus in the setting of  Lemma \ref{lem:explosionC1} with $I=(0,\infty)$ and
$$
s(x) = |a| x, \qquad b(x) = \frac{a_1 \sigma_1 + a_2 \sigma_2}{|a|}x^{\alpha}.
$$
Condition \eqref{eq:condkaratzas} holds because for every $K$ compact interval of $(0,\infty)$ we have
\begin{equation}\notag
\int_K \left( \frac{1}{s^2(y)}+\left|\frac{b(y)}{s(y)}\right|\right)dy = \int_K \left( \frac{1}{|a|y^2}+\left|\frac{a_1 \sigma_1 + a_2 \sigma_2}{|a|^2}y^{\alpha-1}\right|\right)dy<\infty.
\end{equation}
Moreover, $s(\cdot)$ and $b(\cdot)$ are locally H\"older continuous on $(0,\infty)$. The result follows from Lemma \ref{lem:explosionC1}, since $\hat v$ does not reach zero in finite time by Lemma \ref{lemma:sin2}.
\end{proof}
We are now ready to state our result.
\begin{example}\label{ex:problem2}
Consider the solutions $(S,v)$ to the system of SDEs \eqref{eq:S}-\eqref{eq:v}, supposing  $\rho=0$ and $a_1\sigma_1 +a_2\sigma_2>0$. 
Let $\bH$ be the filtration generated by $(B^1, B^2, B^3)$, and $\bF$ the filtration generated by a fourth Brownian motion $B^4$, independent of $(B^1,B^2,B^3)$. Introduce a  $\bF$-local martingale $M$ with volatility term bounded away from zero.

Then the process $X:=M+S$ is an $\bH \cup \bF$-local martingale, and its optional projection into $\bF$ is $\op X = M +m$, where $m_t=S_0P(\{\hat v \text{ does not explode on }[0,t]\})$ by Lemma \ref{lemma:sin1}.  Thus $\op X$ is not an $\bF$-local martingale, because $m$ is not constant by Lemma \ref{lemma:sin2}. However, since the derivative of $m$ is continuous on $(0,\infty)$ by Lemma \ref{lem:mC1}, Corollary \rm{VIII}.1.16 of \citep{RevuzYor} implies that there exists a measure $Q \sim P$ such that $\op X = M +m$ is a $(Q,\bF)$-local martingale. Hence, property \eqref{prob:locmartgen} is satisfied.
\end{example} 

\begin{example}\label{ex:problem3}
Again in the setting and with the notations of Example \ref{ex:problem2}, suppose now $M$ to be a $(P,\bF)$-strict local martingale. The hypothesis that $M$ has volatility bounded away from zero can instead be dropped. Let $Q$ the probability measure from Proposition \ref{prop:sin} under which $S$ is a true martingale. Since the density of $Q$ with respect to $P$ only depends on $B^1,B^2,B^3$, $M$ is a strict local martingale also with respect to $Q$, and $X$ as well. The $Q$-optional projection of $X$ into $\bF$ is given by
$$
\opQ X_t = M_t +\mathbb{E}[S_t]=\opQ X_t = M_t + S_0, \quad t \ge 0,
$$
which is a $(Q,\bF)$-local martingale. 

Thus it holds $P \notin \M^o_{loc}(X, \bF)$ but
$$
 \M_{L}(X,\bH \cup \bF) \cap \M^o_{loc}(X, \bF) \ne \emptyset,
$$
i.e., property \eqref{prob:intersect} holds for $X$.
\end{example}

We now find an example of a sub-filtration $\hat \bF \subset \bG$ such that the optional projection of $S$ into $S$ is not an $\hat \bF$-local martingale. 
The next proposition is a generalization of Lemma 4.2 of \citep{sin1998complications}.
\begin{proposition}\label{prop:stopped}
Suppose that the two-dimensional process $(S,v)$ satisfies the system of SDEs \eqref{eq:S}-\eqref{eq:v}, and call $\bF$ the natural filtration of $B^1$. Introduce the process $(\bar X_t)_{t \ge 0}$ defined by
\begin{equation}\label{eq:hatXsin}
\hat X_t = B_t^1 - \sigma_1\int_0^t v_s^{\alpha}ds, \quad t \ge 0,
\end{equation}
and call $\hat \bF$ the natural filtration of $\hat X$. Then for every $\hat \bF$- stopping time $\hat \tau$ there exists an $\bF$-stopping time $\tau$ such that 
\begin{equation}\label{eq:stoppedattau}
\bE[S_{T\wedge \hat \tau}]=S_0P(\{\hat v \text{ does not explode on }[0,T \wedge \tau]\}),
\end{equation}
where $\hat v$ is defined in \eqref{eq:w}.
\end{proposition}
\begin{proof}
By \eqref{eq:S}, $S$ is a positive local martingale. Define a sequence of stopping times $(\tau_n)_{n \in \bN}$ by
$$
\tau_n=\inf \left\{t \in \reals^+: |\sigma_1+\sigma_2|^2\int_0^t v_s^{2\alpha}ds \ge n\right\} \wedge T,
$$
with $v=(v_t)_{t \ge 0}$ in \eqref{eq:v}.
Then the process $S^n$ defined by 
\begin{equation}\label{eq:Sproof}
S_t^n=S_{t \wedge \tau_n},\quad  t \ge 0
\end{equation}
 is a local martingale for $n \in \bN$. Define $Z^{n}$ by
\begin{equation}\notag
Z_t^{n} = \sigma_1\int_0^{t \wedge \tau_n} v_s^\alpha dB_s^1+\sigma_2\int_0^{t \wedge \tau_n} v_s^\alpha dB_s^2, \quad t \ge 0.
\end{equation}
Then  $S^{n}$ is the stochastic exponential of $Z^{n}$, and since $[Z^{n},Z^{n}]_t \le n$ for all $t \ge 0$, $S^{n}$ is a $(P,\bG)$-martingale for every $n \in \bN$ by Novikov's condition and $(\tau_n)_{n \in \bN}$ reduces $S$ with respect to $(P,\bG)$.

Since $S^n$ stopped at $\hat \tau$ is also a martingale, we can define a new probability measure $Q_n$ on $(\Omega, \G_T)$ as 
$$
Q_n(A)=\frac{1}{S_0}\mathbb{E}[S_{T\wedge \tau_n \wedge \hat \tau}\one_A] \qquad \text{for all $A \in \G_T$}.
$$
By the Lebesgue dominated convergence theorem, it holds
\begin{equation}\label{eq:limitsin}
\mathbb{E}[S_{T  \wedge \hat \tau}]=\lim_{n \to \infty} \mathbb{E}[S_{T\wedge  \tau_n  \wedge \hat \tau }\one_{\{\tau_n \ge  T \wedge  \hat \tau\}}]=S_0\lim_{n \to \infty} Q_n\left(\tau_n \ge T \wedge \hat \tau\right),
\end{equation}
by definition of $Q_n$. Moreover, Girsanov's Theorem implies that the processes $B^{(n,1)}$, $B^{(n,2)}$ defined by
\begin{align}
B_t^{(n,1)}=&B_t^1-\sigma_1\int_0^t \one_{\{s \le \tau_n \wedge \hat \tau\}}v_s^{\alpha}ds, \quad t \ge 0 \label{eq:Btau1} \\
B_t^{(n,2)}=&B_t^2-\sigma_2\int_0^t \one_{\{s \le \tau_n \wedge \hat \tau\}}v_s^{\alpha}ds, \quad t \ge 0 \label{eq:Btau2} 
\end{align}
are Brownian motions under $Q_n$, $n \ge 0$. Therefore under  $Q_n$, the process $v$ has dynamics
\begin{align}
&dv_t=a_1v_tdB_t^{(n,1)}  + a_2v_2dB_t^{(n,2)}  + a_3v_tdB_t^3 + \rho(L-v_t)dt\notag \\
& \qquad +\one_{\{t \le \tau_n \wedge \hat \tau \}}(a_1 \sigma_1+a_2\sigma_2)v^{\alpha+1}_t dt, \quad t \ge 0, \quad v_0=1. \label{eq:newv}
\end{align}
Consider now the process $\hat X$ introduced in  \eqref{eq:hatXsin} and define $\hat v$ as the unique, strong solution of the SDE 
\begin{equation} \label{eq:hatv}
d\hat{v}_t=a_1\hat{v}_tdB^1_t + a_2\hat{v}_tdB^2_t+ a_3v_tdB_t^3+\rho(L-\hat{v}_t)dt + (a_1 \sigma_1+a_2\sigma_2)\hat{v}^{\alpha+1}_t dt,
\end{equation}
$t \ge 0$. Note that on $[0, \tau_n \wedge \hat \tau]$, $(\hat X,v)$ have the same distribution under $Q_n$ as  $(B^1,\hat{v})$ under $P$. 

By the Doob measurability theorem (see, e.g., \cite[Lemma~1.13]{kallenberg2006foundations}), there exists a measurable function $h:\C[0,\infty)\rightarrow\reals^+$ such that $\hat \tau=h(\hat X_\cdot)$. Set $\tau = h(B^1_\cdot)$. As $T \wedge \hat \tau$ is a $\sigma(\hat X)$-stopping time there exists, by the Doob measurability theorem again, a $\B(\C[0,t])$-measurable function $\Psi_t$ such that $\one_{\{t \ge T \wedge \hat \tau\}}=\Psi_t(\hat X^t_{\cdot})$. Thus it holds
\begin{equation}\notag
\one_{\{\tau_n \ge T \wedge \hat \tau\}} = \Psi_{\tau_n}(\hat X^{\tau_n}_{\cdot}), \quad n \in \bN .
\end{equation}
Analogously, by the way we have constructed $\tau$, we have 
\begin{equation}\notag
\one_{\{\hat \tau_n \ge T \wedge \tau\}} = \Psi_{\hat \tau_n}(B^{1,\hat \tau_n}_{\cdot}),\quad n \in \bN,
\end{equation}
where $(\hat{\tau}_n)_{n \in \bN}$, are stopping times for the natural filtration of $\hat v$, defined by
$$
\hat{\tau}_n=\inf \left\{t \in \reals^+: |\sigma_1+\sigma_2|^2\int_0^s \hat{v}_u^{2\alpha}du \ge n\right\}, \quad n \ge 1.
$$
Since on $[0, \tau_n \wedge \hat \tau]$, $(\hat X,v)$ has the same law under $Q_n$ as $(B^1,\hat{v})$ under $P$, we have that 
$\Psi_{\tau_n}(\hat X^{\tau_n}_{\cdot})$ has the same law under $Q_n$ as $\Psi_{\hat \tau_n}(B^{1,\hat \tau_n}_{\cdot})$ under $P$.
Thus, from \eqref{eq:limitsin} we get 
\begin{align}
\mathbb{E}[S_{T \wedge \hat \tau}]&=S_0\lim_{n \to \infty} Q_n\left(\tau_n \ge T \wedge \hat \tau \right) 
\notag \\&=
 S_0\lim_{n \to \infty} \bE^{Q_n}\left[ \Psi_{\tau_n}(\hat X^{\tau_n}_{\cdot}) \right] 
\notag \\&=
S_0\lim_{n \to \infty} \bE^P\left[\Psi_{\hat \tau_n}(B^{1,\hat \tau_n}_{\cdot}) \right] 
\notag \\&=
 S_0\lim_{n \to \infty} P\left(\hat \tau_n \ge T \wedge \tau \right)
 \notag \\&=
 S_0 P\left(\hat\tau_n \ge T \wedge  \tau \text{ for some $n$}\right) 
 \notag \\&= 
 S_0P\left(\hat v \text{ does not explode before time $T \wedge \tau$}\right),\notag
\end{align}
and the proof is complete.
 \end{proof}

We are now ready to give the following

\begin{theorem}\label{thm:nolocalized}
Consider the stochastic volatility process $S$ defined by
\begin{align}
&dS_t=\sigma_1v_t^{\alpha}S_tdB^1_t + \sigma_2v^{\alpha}_tS_tdB_t^2, \quad t \ge 0, \quad S_0 = s>0, \label{eq:S2} \\
&dv_t=a_2v_tdB_t^2 + \rho(L-v_t)dt, \quad t \ge 0, \quad v_0 = 1,\label{eq:v2}
\end{align}
i.e., the model introduced in \eqref{eq:S}-\eqref{eq:v} with $a_1=a_3=0$, and suppose that $a_2\sigma_2>0$.
Consider the filtration $\hat \bF  \subset \bG$, generated by the process $\hat X$ defined in \eqref{eq:hatXsin}. 
Then the $P$-optional projection of $S$ into $\hat \bF$ is not an $\hat \bF$-local martingale.
\end{theorem}
\begin{proof}
The process $S$ in \eqref{eq:S2} is a strict local martingale by Proposition \ref{prop:sin}. By Proposition \ref{prop:stopped}, for every $\hat \bF$-stopping time $\hat \tau$ there exists a $\sigma(B^1)$-stopping time $\tau$ such that  
\begin{equation}\label{eq:Snotlocalized}
\bE[S_{T\wedge \bar \tau}]=S_0P(\{\hat v \text{ does not explode on }[0,T \wedge \tau]\}).
\end{equation}
where $\hat v$ is now given by
\begin{equation}\notag
d\hat v_t=a_2\hat v_tdB_t^2 +\rho(L-\hat v_t)dt + a_2 \sigma_2\hat v^{\alpha+1}_tdt , \quad t \ge 0, \quad \hat v_0 = 1.
\end{equation}
Since $a_1\sigma_1+a_2\sigma_2 = a_2\sigma_2  > 0$, Lemma \ref{lemma:sin2} implies that
$$
P(\{\hat v \text{ does not explode on }[0,t]\})<1 \quad \text{for all $t > 0$}.
$$ 
In particular, 
$$
P(\{\hat v \text{ does not explode on }[0,T \wedge \eta]\})<1
$$
 for every $\sigma(B^1)$-stopping time $\eta$ with $P(\eta = \infty) <1$, because $\hat v$ is independent of $B^1$.
Together with \eqref{eq:Snotlocalized}, this implies that $S$ cannot be localized by any sequence of $\hat \bF$-stopping times. Consequently, the optional projection of $S$ into $\hat \bF$ cannot be an $\hat \bF$-local martingale by Theorem 3.7 of \citep{follmer2011local}.
\end{proof}
By Proposition \ref{prop:sin} and Theorem \ref{thm:nolocalized}, we obtain a further example of two probability measures $P$ and $Q$, of a $P$-local martingale $S$ and of a non trivial filtration $\hat \bF \subset \bG$, such that the optional projection of $S$ into $\hat \bF$ under $P$ is not a $P$-local martingale but the optional  projection of $S$ into $\hat \bF$ under $Q$ is a $Q$-martingale. 


 \appendix
 
\section{Optional projections and optimal transport}\label{app:opttransport}

Consider a probability space $(\Omega, \F, \bP)$ equipped with filtrations $\bF=(\F_t)_{t \ge 0}$ and $\bG=(\G_t)_{t \ge 0}$, $\bF \subset \bG$, both satisfying the usual hypothesis of right-continuity and completeness. 

For $\bG$-adapted c\`adl\`ag processes $X$ and $Z$, we denote $X\ll_\bG Z$ if $Z-X$ is a nonnegative $\bG$-supermartingale. 
\begin{proposition}\label{prop:locmartsupermart}
Let $X$ be a nonnegative, c\`adl\`ag $\bG$-supermartingale. Then $X$ is a $\bG$-local martingale if and only if $X\ll_\bG Z$ for all $\bG$- supermartingales $Z\ge X$. 
\end{proposition}
\begin{proof}
Assume that $X$ is a local martingale, and consider a supermartingale $Z \ge X$. Let $(\tau_n)_{n \ge 0}$ be a localizing sequence for $X$. By Fatou's lemma, 
\[
\bE[Z_t-X_t | \F_s] \le \liminf_{n\to \infty} \bE[Z_{t \wedge \tau_n} - X_{t \wedge \tau_n} | \F_s]  \le Z_s -  \liminf_{n\to \infty} \bE[ X_{t \wedge \tau_n} | \F_s] = Z_s - X_s
\]
for every $0 \le s \le t$, so $Z-X$ is a supermartingale.

Suppose now that $Z-X$ is a nonnegative supermartingale for every supermartingale $Z \ge X$. Assume also that $X$ is not a local martingale, i.e., that it is a strict supermartingale. Then $X$ has the Doob-Meyer decomposition 
\[
X_t = M_t - Y_t, \quad t \ge 0,
\]
where  $M \ge X$ is a nonnegative local martingale and $Y\ne 0$ is a nondecreasing  process. Thus $M\ge X$ is a supermartingale for which $M-X = Y$ is not a supermartingale, which is a contradiction.
\end{proof}

Proposition \ref{prop:locmartsupermart} can be used to characterize when the optional projection of a local martingale remains a local martingale, and to provide a sufficient and necessary condition for this property via optimal transport.

\begin{theorem}
For a nonnegative $\bG$-local martingale $X$,  $\op X$ is a $\bF$-local martingale if and only if, for every $\bF$-supermartingale $Y\ge \op X$,  there is a $\bG$-supermartingale $Z\ge X$ with $\op Z \ll_\bF Y$.
\end{theorem}
\begin{proof}
If $\op X$ is a local martingale, we may choose $Z=X$. To prove the converse, $Z-X$ is a nonnegative $\bG$-supermartingale, $\op Z-\op X$ is a nonnegative $\bF$-supermartingale, so $Y-\op X = Y-\op Z+\op Z-\op X$ is a nonnegative $\bF$-supermartingale.  
\end{proof}

Given laws $\nu_s,\nu_t$ on $\mathbb R_+$, we denote $\nu_s \ll_{cdo} \nu_t$ if $\int f d \nu_s \le \int f d \nu_t$ for all $f\in C_d$, where $C_d$ is the set of real-valued convex and decreasing functions on $\reals_+$. Let now $X$ be a $\bG$-adapted nonnegative process with law $\nu$ and $\nu_s$ be the law of $X_s$, $s \ge 0$. If $\nu_s \ll_{cdo} \nu_t$ for all $s \le t$,  we say that the law $\nu$ is convex decreasing. In this case, there exists a Markov process with law $\nu$ which is a $\bG$-supermartingale; see Theorem 3 of \citet{kellerer1972markov}. Also note that if a process is a $\bG$-supermartingale, its law is convex decreasing.

We denote by $S(\nu)$ the set of joint laws of $(Z,X)$, where $Z$ ranges over all supermartingales dominating $X$. For $\pi \in S(\nu)$, $\pi_t$ denotes the joint law of $(Z_t,X_t)$. An application of Proposition \ref{prop:locmartsupermart} leads to the following result.

\begin{proposition}
Given a nonnegative $\bG$-adapted process $X$ with the law $\nu$, $X$ is a $\bG$-local martingale if and only if for every $ t \ge 0$,
\begin{equation}\label{eq:conditionlaws}
\sup_{\pi \in S(\nu)} \sup_{s< t} \sup_{f\in C_d} \left[\int f(z -x)d\pi_t(z,x) - \int f(z -x)d\pi_s(z,x) \right]\le 0.
\end{equation}
\end{proposition}
\begin{proof}
If \eqref{eq:conditionlaws} does not hold, there exists a supermartingale $Z \ge X$ such that the law of $Z - X$ is not convex decreasing. Then $Z - X$ is not a supermartingale, so $X$ is not a local martingale by Proposition \ref{prop:locmartsupermart}. Suppose now that $X$ is not a local martingale. By Proposition \ref{prop:locmartsupermart}, there exists a supermartingale $Z \ge X$ such that $Z-X$ is not a supermartingale. Then the law of $Z-X$ is not convex decreasing, and \eqref{eq:conditionlaws} fails.
\end{proof}

\bibliography{bib}

\begin{thebibliography}{}

\bibitem[Biagini et~al., 2014]{Biagini}
Biagini, F., F{\"o}llmer, H., and Nedelcu, S. (2014).
\newblock {Shifting martingale measures and the slow birth of a bubble}.
\newblock {\em Finance and Stochastics}, 18(2):297--326.

\bibitem[Cetin et~al., 2004]{cetin2004modeling}
Cetin, U., Jarrow, R., Protter, P., and Y{\i}ld{\i}r{\i}m, Y. (2004).
\newblock Modeling credit risk with partial information.
\newblock {\em The Annals of Applied Probability}, 14(3):1167--1178.

\bibitem[Cox and Hobson, 2005]{CoxHobson}
Cox, A. and Hobson, D. (2005).
\newblock Local martingales, bubbles and option prices.
\newblock {\em Finance Stochastics}, 9(4):477--492.

\bibitem[Delbaen and Schachermayer, 1994]{delbaen1994arbitrage}
Delbaen, F. and Schachermayer, W. (1994).
\newblock Arbitrage and free lunch with bounded risk for unbounded continuous
  processes.
\newblock {\em Mathematical Finance}, 4(4):343--348.

\bibitem[F{\"o}llmer and Protter, 2011]{follmer2011local}
F{\"o}llmer, H. and Protter, P. (2011).
\newblock Local martingales and filtration shrinkage.
\newblock {\em ESAIM: Probability and Statistics}, 15:S25--S38.

\bibitem[Jarrow et~al., 2011]{JarrowKchiaProtter}
Jarrow, R., Kchia, Y., and Protter, P. (2011).
\newblock { How to detect an asset bubble}.
\newblock {\em SIAM Journal on Financial Mathematics}, 2:839--865.

\bibitem[Jarrow and Protter, 2004]{jarrow2004structural}
Jarrow, R. and Protter, P. (2004).
\newblock Structural versus reduced form models: a new information based
  perspective.
\newblock {\em Journal of Investment management}, 2(2):1--10.

\bibitem[Jarrow and Protter, 2013]{jarrow2013positive}
Jarrow, R. and Protter, P. (2013).
\newblock Positive alphas, abnormal performance, and illusory arbitrage.
\newblock {\em Mathematical Finance: An International Journal of Mathematics,
  Statistics and Financial Economics}, 23(1):39--56.

\bibitem[Jarrow et~al., 2007a]{JarrowProtter2007}
Jarrow, R., Protter, P., and Shimbo, K. (2007a).
\newblock {Asset price bubbles in complete markets}.
\newblock {\em Advances in Mathematical Finance}, In Honor of Dilip B.
  Madan:105--130.

\bibitem[Jarrow et~al., 2010]{JarrowProtter2010}
Jarrow, R., Protter, P., and Shimbo, K. (2010).
\newblock {Asset price bubbles in incomplete markets}.
\newblock {\em Mathematical Finance}, 20(2):145--185.

\bibitem[Jarrow et~al., 2007b]{jarrow2007information}
Jarrow, R.~A., Protter, P., and Sezer, A.~D. (2007b).
\newblock Information reduction via level crossings in a credit risk model.
\newblock {\em Finance and Stochastics}, 11(2):195--212.

\bibitem[Kallenberg, 2006]{kallenberg2006foundations}
Kallenberg, O. (2006).
\newblock {\em Foundations of modern probability}.
\newblock Springer Science \& Business Media.

\bibitem[Karatzas and Ruf, 2016]{karatzas2016distribution}
Karatzas, I. and Ruf, J. (2016).
\newblock Distribution of the time to explosion for one-dimensional diffusions.
\newblock {\em Probability Theory and Related Fields}, 164(3-4):1027--1069.

\bibitem[Kellerer, 1972]{kellerer1972markov}
Kellerer, H.~G. (1972).
\newblock Markov-komposition und eine anwendung auf martingale.
\newblock {\em Mathematische Annalen}, 198(3):99--122.

\bibitem[Larsson, 2014]{larsson2014filtration}
Larsson, M. (2014).
\newblock Filtration shrinkage, strict local martingales and the f{\"o}llmer
  measure.
\newblock {\em The Annals of Applied Probability}, 24(4):1739--1766.

\bibitem[Loewenstein and Willard, 2000]{LoewensteinWillard}
Loewenstein, M. and Willard, G. (2000).
\newblock Rational equilibrium asset-pricing bubbles in continuous trading
  models.
\newblock {\em Journal of Economic Theory}, 91(1):17--58.

\bibitem[Mijatovic and Urusov, 2012]{MijatovicUrusov}
Mijatovic, A. and Urusov, M. (2012).
\newblock On the martingale property of certain local martingales.
\newblock {\em Probability Theory and Related Fields}, 152(1-2):1--30.

\bibitem[Protter, 2013]{Protter2013}
Protter, P. (2013).
\newblock {\em {A mathematical theory of financial bubbles}}, volume 2081 of
  Lecture Notes in Mathematics of {\em {V. Henderson and R. Sincar editors,
  Paris-Princeton Lectures on Mathematical Finance}}.
\newblock Springer.

\bibitem[Revuz and Yor, 1999]{RevuzYor}
Revuz, D. and Yor, M. (1999).
\newblock {\em Continuous Martingales and Brownian Motion, Third Edition}.
\newblock Springer-Verlag, New York.

\bibitem[Sin, 1998]{sin1998complications}
Sin, C.~A. (1998).
\newblock Complications with stochastic volatility models.
\newblock {\em Advances in Applied Probability}, 30(1):256--268.

\bibitem[Yang and Chu, 2017]{yang2017approximating}
Yang, Z.-H. and Chu, Y.-M. (2017).
\newblock On approximating the modified bessel function of the second kind.
\newblock {\em Journal of Inequalities and Applications}, 2017(1):41.

\end{thebibliography}
\end{document}